\documentclass[runningheads,envcountsame]{llncs}
\usepackage[T1]{fontenc}
\usepackage{graphicx}
\usepackage{hyperref}
\usepackage{color}

\usepackage{booktabs}
\usepackage{amssymb}
\usepackage{amsmath}
\usepackage[capitalize]{cleveref}
\usepackage{enumerate}
\usepackage{ebproof}
\usepackage{stmaryrd}
\usepackage{xcolor}
\usepackage{url}
\usepackage{xspace}
\usepackage{todonotes}
\usepackage{cancel}
\usepackage{bbding}

\renewcommand{\emptyset}{\varnothing}

\newcommand{\m}[1]{\mathsf{#1}}

\newcommand{\mb}[1]{\mathsf{#1}}
\newcommand{\seq}[2][n]{{#2_1},\dots,{#2_{#1}}}
\renewcommand{\vec}{\boldsymbol}

\newenvironment{psmallmatrix}
 {\left(\begin{smallmatrix}}
 {\end{smallmatrix}\right)}

\newcommand{\NN}{\mathbb{N}}
\newcommand{\ZZ}{\mathbb{Z}}
\renewcommand{\AA}{\mathcal{A}}
\newcommand{\BB}{\mathcal{B}}
\newcommand{\CC}{\mathcal{C}}
\newcommand{\DD}{\mathcal{D}}
\newcommand{\FF}{\mathcal{F}}
\newcommand{\HH}{\mathcal{H}}

\newcommand{\MM}{\mathcal{M}}
\newcommand{\OO}{\mathcal{O}}
\newcommand{\PP}{\mathcal{P}}
\newcommand{\RR}{\mathcal{R}}
\renewcommand{\SS}{\mathcal{S}}
\newcommand{\TT}{\mathcal{T}}
\newcommand{\VV}{\mathcal{V}}

\newcommand{\rt}{\mathrm{root}}

\newcommand{\rto}{\xrightarrow{\smash{\epsilon}}}

\newcommand{\superterm}{\trianglerighteqslant}
\newcommand{\prsuperterm}{\rhd}
\newcommand{\nprsuperterm}{\ntriangleright}

\newcommand{\sqsupsetsim}{\mathrel{\vphantom{\gtrsim}\smash{\ooalign{\raise.4ex\hbox{$\sqsupset$}\cr$\raise-.85ex\hbox{$\sim$}$}}}}

\newcommand{\lex}{\mathsf{lex}}
\newcommand{\kbo}{\mathsf{kbo}}

\newcommand{\rpo}{\mathsf{rpo}}
\newcommand{\wpo}{\mathsf{wpo}}

\newcommand{\DP}{\mathsf{DP}}
\newcommand{\Emb}{{\mathcal{E}\mathsf{mb}}}

\newcommand{\emptylist}{[\,]}

\newcommand{\NaTT}{\textsf{NaTT}}

\begin{document}

\title{Lexicographic Combination of Reduction Pairs (Extended Version)}

\author{Teppei Saito\textsuperscript{(\Envelope)}\orcidID{0009-0001-9786-0044}\and
Nao Hirokawa\textsuperscript{(\Envelope)}\orcidID{0000-0002-8499-0501}}

\authorrunning{T. Saito and N. Hirokawa}

\institute{JAIST, Nomi, Japan
\email{\{saito,hirokawa\}@jaist.ac.jp}}

\maketitle

\begin{abstract}
We present a simple criterion for combining reduction pairs
lexicographically. The criterion is applicable to arbitrary classes
of reduction pairs,
such as the polynomial interpretation, the matrix interpretation, and the
Knuth--Bendix order. In addition, we investigate a variant of the matrix
interpretation where the lexicographic order is employed instead of the usual
component-wise order.
Effectiveness is demonstrated by experiments and examples,
including Touzet's Hydra Battle.
\end{abstract}

\section{Introduction}

\emph{Lexicographic combination} is a powerful method to combine termination
measures into a more complex one. To illustrate it, consider the term rewrite
system
\begin{align*}
1\colon\quad \m{f}(\m{f}(x)) & \to \m{g}(\m{g}(\m{f}(x)))
&
2\colon\quad \m{g}(\m{g}(x)) & \to x
\end{align*}
which has the following rewrite sequence:
\[
\m{f}(\m{f}(\m{f}(x)))
\xrightarrow{1}
\m{f}(\m{g}(\m{g}(\m{f}(x))))
\xrightarrow{2}
\m{f}(\m{f}(x))
\xrightarrow{1}
\m{g}(\m{g}(\m{f}(x)))
\xrightarrow{2}
\m{f}(x)
\]
If the numbers of occurrences of $\m{f}$ and $\m{g}$ are measured,
the sequence turns into the descending sequence with respect to the
lexicographic order $>^\lex$.
\[
(3,0) >^\lex
(2,2) >^\lex
(2,0) >^\lex
(1,2) >^\lex
(1,0)
\]
In this manner, lexicographic combination can be used for showing
termination of the rewrite system.
More formally speaking, the measure can be expressed as the lexicographic
combination of two linear polynomial interpretations~\cite{L79}, namely $\AA$ and $\BB$
defined by
$\m{f}_\AA(x) = x + 1$, $\m{g}_\AA(x) = x$,
$\m{f}_\BB(x) = x$, and $\m{g}_\BB(x) = x + 1$. 
Here $\AA$ counts the number of occurrences of $\m{f}$, while $\BB$ counts that of $\m{g}$.

There has been a long line of research concerning termination analysis of
term rewrite systems.
Among others, the dependency pair
framework~\cite{AG00,GAO02,GTS05,HM05,HM07} is a powerful method
for automated termination analysis.  Decreasing
measures for the method are typically given in the form of 
\emph{reduction pairs} consisting of preorders and well-founded orders on terms.
Unfortunately, in general,
lexicographic combinations of reduction pairs are not reduction
pairs. To overcome this, we give a simple criterion for a
combination to be a reduction pair.

Actually, the previous example implicitly uses the folklore that
lexicographic combinations of \emph{monotone} reduction pairs are
(monotone) reduction pairs, see \cite{BL87,G90,TGS04}.
Our main result can be conceived as an extension of it.
For example, let us consider the following term rewrite system
\begin{align*}
\m{s}(x) + y & \to \m{p}(\m{s}(x)) + \m{s}(y)
&
\m{p}(\m{s}(x)) & \to x
\end{align*}
which encodes addition in a tricky way, using the successor symbol $\m{s}$ and the predecessor symbol $\m{p}$.
The dependency pair framework tells us that the termination is established
if one can find a reduction pair $({\geqslant}, {>})$ that fulfills the
following set of constraints.
\begin{align*}
\m{s}(x) + y & \geqslant \m{p}(\m{s}(x)) + \m{s}(y)
&
\m{p}(\m{s}(x)) & \geqslant x
\\
\m{s}(x) +^\sharp y & > \m{p}(\m{s}(x)) +^\sharp \m{s}(y)
&
\m{s}(x) +^\sharp y & > \m{p}^\sharp(\m{s}(x))
\end{align*}
Here, $+^\sharp$ and $\m{p}^\sharp$ are fresh symbols introduced by the method.
To do this, one can use the lexicographic combination of the following linear polynomial interpretations $\AA$ and $\BB$:
\begin{align*}
\m{s}_\AA(x) &= x + 1
&
\m{p}_\AA(x) &= x
&
\m{p}^\sharp_\AA(x) &= 0
&
x +_\AA y &= x
&
x +^\sharp_\AA y &= x
\\
\m{s}_\BB(x) &= x + 1
&
\m{p}_\BB(x) &= 0
&
\m{p}^\sharp_\BB(x) &= 0
&
x +_\BB y &= x
& 
x +^\sharp_\BB y &= x
\end{align*}
The resulting lexicographic combination satisfies the 
constraints: Thanks to lexicographic comparison, the constraint
$\m{p}(\m{s}(x)) \geqslant x$ is not subject to comparison by the second
interpretation $\BB$, as it is interpreted to 
$x + 1 > x$ in the first interpretation $\AA$.
This allows us to use $\m{p}_\BB(x) = 0$,
which helps to satisfy the remaining constraints.
A problem here is that, in contrast to the previous example, the algebras $\AA$ and $\BB$ are not monotone,
as $\m{p}_\BB(x) = 0$ disregards the argument $x$.
Although the folklore does not apply,
our result can justify that the resulting combination is indeed a reduction pair,
which facilitates a successful termination proof.

The polynomial interpretation~\cite{L79}, the matrix
interpretation~\cite{EWZ08},
and the Knuth--Bendix order~\cite{KB70}
are typical methods for constructing reduction pairs.
We also show how to use these methods with lexicographic combination,
demonstrating it with the term rewrite system of the Battle of Hercules and
Hydra~\cite{KP82} due to Touzet~\cite{T98}.
As a by-product,
we obtain a variant of the matrix interpretation with the lexicographic order
instead of the standard component-wise order.
As remarked in~\cite{EWZ08,NM11}, this has been an open question.

In the context of the dependency pair framework, repetitive application of
reduction pair processors~\cite{GTS05,HM05} and the rule removal
method~\cite{TGS04} can be regarded as alternative methods for combining
reduction pairs in a lexicographic manner.  Examples and experimental data
show that our method is complementary to those methods
and particularly useful for relative termination.

This paper is an extended version of \cite{SH25} with appendices.

\paragraph{The structure of the paper.}
After recalling some preliminaries in \cref{sec:preliminaries}, a criterion
for lexicographic combination is presented
in \cref{sec:criterion}.  The termination proof of the Hydra Battle is discussed
in \cref{sec:hydra}.
Then, as a generalization of lexicographic combination of the linear polynomial
interpretation, \cref{sec:matrix} studies a variant of the matrix interpretation
with the lexicographic order.
Experimental data and related work are discussed in \cref{sec:experiments}
and \cref{sec:conclusion}, respectively.

\section{Preliminaries}
\label{sec:preliminaries}

Throughout the paper, we assume familiarity with term
rewriting~\cite{BN98,TeReSe}.  In this section, we briefly recall
notions and notations for term rewriting and termination analysis based
on the dependency pair framework.

Let $\FF$ be a signature and $\VV$ an infinite set of variables.  The
set of terms built from $\FF$ and $\VV$ is denoted by $\TT(\FF, \VV)$, or
simply by $\TT$.  
The root symbol $\rt(t)$ of a non-variable term $t = f(\seq{t})$ is $f$.
A mapping $\sigma : \VV \to \TT$ is called a
\emph{substitution} if $\sigma(x) \neq x$ for only finitely many variables
$x$.  Given a term $t$, the term obtained by replacing each variable
occurrence $x$ in $t$ with $\sigma(x)$ is denoted by $t\sigma$.  
Let $\square$ be a fresh constant symbol. A
\emph{context} is a term with exactly one occurrence of $\square$.
The term obtained by replacing $\square$ in a context $C$ with a term
$t$ is denoted by $C[t]$.  We write $s \superterm t$ if $s = C[t]$ for some
context $C$, and moreover we write $s \prsuperterm t$ if $s \superterm t$
and $s \neq t$.  Let $\hookrightarrow$ be a relation on terms.  The
relation $\hookrightarrow$ is said to be \emph{closed under substitutions}
if $s\sigma \hookrightarrow t\sigma$ whenever $s \hookrightarrow t$ and
$\sigma$ is a substitution.  Similarly, $\hookrightarrow$ is said to be
\emph{closed under contexts} if $C[s] \hookrightarrow C[t]$ whenever $t
\hookrightarrow u$ and $C$ is a context.

A \emph{rewrite rule} $\ell \to r$ is a pair $(\ell, r)$ of terms such that
$\ell$ is not a variable and every variable occurring in $r$ also occurs in
$\ell$.  A set of rewrite rules is called a \emph{term rewrite system}
(TRS).  The \emph{rewrite step} $\to_\RR$ of a TRS $\RR$ is defined as
follows: $s \to_\RR t$ if $s = C[\ell\sigma]$ and $t = C[r\sigma]$ for
some rewrite rule $\ell \to r \in \RR$, context $C$, and substitution $\sigma$.
When $C = \Box$, it is written as $s \rto_\RR t$.
Given a term $t$ and a variable $x$, we write $|t|_x$ for the number
of occurrences of $x$ in $t$.  A term rewrite system $\RR$ is
\emph{non-duplicating} if $|\ell|_x \geqslant |r|_x$ for all 
$\ell \to r \in \RR$ and variables $x$.
A term $t$ is \emph{terminating} with respect to a relation
$\hookrightarrow$ on terms if there is no infinite sequence
$t \hookrightarrow t_1 \hookrightarrow t_2 \hookrightarrow \cdots$
starting from $t$.
\emph{Termination} of the relation $\hookrightarrow$ is defined as absence
of non-terminating terms.
A TRS $\RR$ is \emph{terminating} if $\to_\RR$ is.  
Let $\hookrightarrow^*$ denote the reflexive and transitive closure of
a relation $\hookrightarrow$.
Given TRSs $\RR$ and $\SS$, the relation 
$\to_{\RR/\SS}$ is defined on terms as follows: $s \to_{\RR/\SS} t$ if
$s \to_\SS^* u \to_\RR v \to_\SS^* t$ for some terms $u$ and $v$.
We say that $\RR$ is \emph{relatively terminating} with respect to $\SS$,
or $\RR/\SS$ is terminating, if $\to_{\RR/\SS}$ is terminating.

We recall the dependency pair
framework~\cite{AG00,GTS05,HM05,HM07}.
Let $\RR$ be a TRS.  We define the set $\DD_\RR$ of \emph{defined symbols} as 
$\{ f \mid f(\seq{\ell}) \to r \in \RR \}$.
Given a term $t$ of the form $f(\seq[n]{t})$
with $f \in \DD_\RR$, we write $t^\sharp$ for $f^\sharp(\seq[n]{t})$.
Here $f^\sharp$ is a fresh $n$-ary function symbol corresponding to $f$.  
The set of such terms $t^\sharp$ is denoted by $\TT^\sharp$.
The TRS $\DP(\RR)$ is defined as follows:
\[
\DP(\RR) =
\{ \ell^\sharp \to t^\sharp \mid \text{$\ell \to r \in \RR$, 
$r \superterm t$, $\rt(t) \in \DD_\RR$, and $\ell \nprsuperterm t$} \}
\]
Rules in $\DP(\RR)$ are called \emph{dependency pairs}.
\emph{Dependency pair problems} are pairs $(\PP,\RR)$ of TRSs with
$\PP \subseteq \TT^\sharp \times \TT^\sharp$ and
$\RR \subseteq \TT \times \TT$.
A dependency pair problem $(\PP,\RR)$ is \emph{finite} if
there exists no infinite sequence of the form
$s_1 \to_\RR^* t_1 \rto_\PP s_2 \to_\RR^* t_2 \rto_\PP \cdots$,
where each $s_i$ is terminating with respect to $\RR$.

\begin{theorem}
\label{thm:dp}
A TRS $\RR$ is terminating if and only if $(\DP(\RR),\RR)$ is finite.
\end{theorem}

A pair of a preorder $\geqslant$ and a strict order $>$ on the same set is
an \emph{order pair} if $a > d$ whenever $a \geqslant b > c \geqslant d$.
An order pair is \emph{well-founded} if $>$ is well-founded.  An order pair
$({\geqslant},{>})$ on terms is \emph{stable} if $\geqslant$ and $>$ are
closed under substitutions.  A well-founded stable order pair is a
\emph{reduction pair} if $\geqslant$ is closed under contexts.  If in
addition $>$ is closed under contexts, $({\geqslant},{>})$ is called a
\emph{monotone} reduction pair.  The next theorem is known as the reduction
pair processor.

\begin{theorem}
\label{thm:rp}
Let $({\geqslant},{>})$ be a reduction pair.  A dependency pair problem
$(\PP,\RR)$ with $\PP \cup \RR \subseteq {\geqslant}$ is finite if and only
if $(\PP \setminus {>},\RR)$ is finite.
\end{theorem}

Thus, a TRS $\RR$ is terminating if $\DP(\RR) \subseteq {>}$ and $\RR
\subseteq {\geqslant}$ for some reduction pair $({\geqslant},{>})$.  
This simple criterion can also be used for showing relative termination.
We say that a TRS $\RR$ \emph{dominates} a TRS $\SS$ if $r$ has no defined
symbols of $\RR$ for all rules $\ell \to r \in \SS$.

\begin{theorem}[\textnormal{\cite{INVY17}}]
\label{thm:dp-relative}
Let $\RR$ and $\SS$ be TRSs such that $\RR$ dominates $\SS$ and $\SS$
is non-duplicating.  If there exists a reduction pair $({\geqslant}, {>})$
with $\DP(\RR) \subseteq {>}$ and $\RR \cup \SS \subseteq {\geqslant}$ then
$\RR/\SS$ is terminating.
\end{theorem}

Let $\FF$ be a signature. An $\FF$-\emph{algebra} $\AA$ is a pair
$(A, {\{f_\AA\}_{f \in \FF}})$ where $A$ is a non-empty set,
called the \emph{carrier} of $\AA$, and each $f_\AA$ is an $n$-ary function
on $A$, called the interpretation of an $n$-ary function symbol $f$.  
Let $\AA$ be an $\FF$-algebra.  A function $\alpha$ from $\VV$ to $A$ is
called an \emph{assignment} for $\AA$. It is extended to the homomorphism 
$[\alpha]_\AA : \TT(\FF,\VV) \to A$
as follows:
\[
[\alpha]_\AA(x) =
\begin{cases}
\alpha(x) & \text{if $x$ is a variable} \\
f_\AA([\alpha]_\AA(t_1), \ldots, [\alpha]_\AA(t_n))
& \text{if $t = f(\seq{t})$} 
\end{cases}
\]
Assume that $\AA$ is equipped with an order pair $({\geqslant},{>})$ on
$A$.  The algebra $\AA$ is \emph{well-founded} if $>$ is well-founded, and
\emph{weakly monotone} if 
\(
f_\AA(a_1, \ldots, a_i, \ldots, a_n) \geqslant
f_\AA(a_1, \ldots, b, \ldots, a_n)
\)
whenever $f \in \FF$ and $a_i \geqslant b$.  We write $s \geqslant_\AA t$ and 
$s >_\AA t$ if $[\alpha]_\AA(s) \geqslant [\alpha]_\AA(t)$ and
$[\alpha]_\AA(s) > [\alpha]_\AA(t)$ hold for all assignments $\alpha$,
respectively. It is known that $({\geqslant_\AA},{>_\AA})$ is a reduction
pair if $\AA$ is weakly monotone and well-founded.

The \emph{matrix interpretation}~\cite{EWZ08} provides a semantic method to
construct reduction pairs.  The carrier of a matrix interpretation $\AA$ is
the set of vectors $\vec{x}$ of natural numbers of a fixed dimension 
$d > 0$. Vectors are ordered by the component-wise order pair $(\geqslant,
>)$ defined as follows:
$(x_1, \ldots, x_d)^T \geqslant (y_1, \ldots, y_d)^T$ if $x_i \geqslant y_i$ for all $i \in \{ 1, \ldots, d \}$;
if in addition $x_1 > y_1$ then $(x_1, \ldots, x_d)^T > (y_1, \ldots, y_d)^T$.
We write 
$(A)_{i,j}$ or simply $A_{i,j}$ for the entry at the $i$-th row and the $j$-th column.
Moreover, $\vec{e}_i$ stands for the unit vector $(0, \ldots, 1, \ldots, 0)^T$
having $1$ only at the $i$-th coordinate, and
$O$ (resp. $\vec{0}$) for the zero matrix (resp. vector) with $0$ in all entries.
The interpretation $f_\AA$ of each $n$-ary function symbol $f$ is of the form
\[
f_\AA(\vec{x}_1, \ldots, \vec{x}_n) = A_1 \vec{x}_1 + \ldots + A_n \vec{x}_n + \vec{a}
\]
where $A_1, \ldots, A_n$ are $d \times d$ matrices of natural numbers and $\vec{a} \in \NN^d$.
The pair $({\geqslant_\AA}, {>_\AA})$ is a reduction pair.
The special class of the matrix interpretation with $d = 1$ is called the \emph{linear polynomial interpretation}.

The \emph{Knuth--Bendix order} (KBO)~\cite{KB70} is another way to
construct reduction pairs. A \emph{weight function} is a pair
of a positive integer $w_0$ and a function $w$ from function symbols to
natural numbers
with $w(c) \geqslant w_0$ for all constants $c$.
The \emph{weight} $w(t)$ of a term $t$ is defined inductively:
$w(x) = w_0$ for variables $x$, and $w(t) = w(f) + \sum_{i} w(t_i)$
for $t = f(t_1, \ldots, t_n)$.
We write $f^n(t)$ for the $n$-times application of a unary function symbol $f$ to a term $t$.
Let $(w_0, w)$ be a weight function and $\succ$ a precedence, that is, a strict order on function symbols.
The Knuth--Bendix order $>_\kbo$ is inductively defined as follows: $s >_\kbo t$ if
$|s|_x \geqslant |t|_x$ for all variables $x$, and
\begin{enumerate}[1.]
\item $w(s) > w(t)$, or
\item $w(s) = w(t)$ and one of the following conditions holds.
\begin{enumerate}[a.]
    \item $s = f^n(t)$ for some $n > 0$ and $t$ is a variable.
    \item $s = f(\seq[m]{s})$, $t = g(\seq[n]{t})$, and either
    \begin{enumerate}[i.]
        \item $f \succ g$ or
        \item $f = g$ and there is an $i$ such that $s_i >_\kbo t_i$ and $s_j = t_j$ for all $j < i$.
    \end{enumerate}
\end{enumerate}
\end{enumerate}
The weight function $(w_0, w)$ is \emph{admissible} for $\succ$,
if $f \succ g$ for other function symbols $g$
whenever $f$ is a unary function symbol with $w(f) = 0$.
If $\succ$ is well-founded and $(w_0, w)$ is admissible for $\succ$,
then $({\geqslant_\kbo}, {>_\kbo})$ is a monotone reduction pair,
where $\geqslant_\kbo$ stands for the reflexive closure of $>_\kbo$.
(We note that in that case $>_\kbo$ is a so-called reduction order.)

\emph{Argument filtering}~\cite{AG00} is a popular transformation for
building reduction pairs with the Knuth--Bendix order.
As the name suggests, this transformation
filters out arguments from a given term. Formally,
an \emph{argument filter} $\pi$ is a mapping that 
associates
each $n$-ary function symbol $f$ to an integer $i$ or a list
$[i_1,\ldots,i_m]$ of integers over $\{1,\ldots,n\}$.  
The \emph{argument filtering} $\hat{\pi}$ is defined as follows:
\[
\hat{\pi}(t) =
\begin{cases}
t & \text{if $t$ is a variable}
\\
\hat{\pi}(t_i)
& \text{if $t = f(t_1, \ldots, t_n)$ and $\pi(f) = i$}
\\
f(\hat{\pi}(t_{i_1}),\ldots,\hat{\pi}(t_{i_m}))
& \text{if $t = f(\seq{t})$ and $\pi(f) = [\seq[m]{i}]$}
\end{cases}
\]
Note that arities of function symbols may change by applying
$\hat{\pi}$.  Given a binary relation $R$ on terms, we write 
$s \mathrel{R}^\pi t$ if $\hat{\pi}(s) \mathrel{R} \hat{\pi}(t)$.

\begin{proposition}
\label{prop:af}
Let $({\geqslant}, {>})$ be a reduction pair and $\pi$ an argument filter.
Then $({\geqslant^\pi}, {>^\pi})$ is a reduction pair.
\end{proposition}

\section{Combinability Criterion}
\label{sec:criterion}

Lexicographic combination is a well-known method to turn two order pairs
into a single order pair which captures a more complicated
termination measure.

\begin{definition}
\label{def:lex-comb}
Let $({\geqslant_1}, {>_1})$ and $({\geqslant_2}, {>_2})$ be 
order pairs on a set $X$.  The \emph{lexicographic combination}
$({\geqslant_{12}}, {>_{12}})$ is the pair of relations on $X$ defined
as follows:
\begin{itemize}
\item
$x \geqslant_{12} y$ if $x >_1 y$, or both $x \geqslant_1 y$ and
$x \geqslant_2 y$.
\item
$x >_{12} y$ if $x >_1 y$, or both $x \geqslant_1 y$ and $x >_2 y$.
\end{itemize}
\end{definition}

Lexicographic combinations of reduction pairs satisfy
all conditions to be reduction pairs,
except for closure under contexts of preorders.

\begin{example}
\label{ex:lex-comb-redpair-invalid}
Consider the reduction pairs $({\geqslant_\AA}, {>_\AA})$ and $({\geqslant_\BB}, {>_\BB})$ induced by the following
linear polynomial interpretations $\AA$ and $\BB$:
\begin{align*}
\m{f}_\AA(x) &= 0
&
\m{a}_\AA &= 1
&
\m{b}_\AA &= 0
\\
\m{f}_\BB(x) &= x
&
\m{a}_\BB &= 0
&
\m{b}_\BB &= 1
\end{align*}
Although $\m{a} \geqslant_{\AA \BB} \m{b}$ follows from $\m{a} >_\AA
\m{b}$, the desired inequality $\m{f}(\m{a}) \geqslant_{\AA \BB} \m{f}(\m{b})$ for closure under contexts
 does not hold.
Even worse, the flipped inequality
$\m{f}(\m{b}) >_{\AA \BB} \m{f}(\m{a})$ follows from
$\m{f}(\m{b}) \geqslant_\AA \m{f}(\m{a})$ and
$\m{f}(\m{b}) >_\BB \m{f}(\m{a})$.
\end{example}

As noted in~\cite{ZWM15}, the problem in
\cref{ex:lex-comb-redpair-invalid} is that, in the first component
$\m{f}_\AA(x) = 0$ is not monotone with respect to the argument $x$,
while $\m{f}_\BB(x) = x$ in the second component is dependent on $x$.
This observation suggests the following definitions.

\begin{definition}
\label{def:monotone-invariant}
Let $({\geqslant}, {>})$ be an order pair on terms.
Let $f$ be an $n$-ary function symbol.  The $i$-th argument position of
$f$ is $>$-\emph{monotone} (or monotone with respect to $>$) if 
\[
t_i > u \implies
f(t_1,\ldots,t_i,\ldots,t_n) > f(t_1,\ldots,u,\ldots,t_n)
\]
for all terms $\seq{t},u$.  Similarly, the position is $\geqslant$-\emph{invariant} if
\[
f(t_1,\ldots,t_i,\ldots,t_n) \geqslant f(t_1,\ldots,u,\ldots,t_n)
\]
for all terms $\seq{t},u$.  
An order pair $({\geqslant},{>})$ is said to be \emph{normal} if ${>} \subseteq
{\geqslant}$ holds. We say that an order pair $({\geqslant_1},{>_1})$ on
terms is \emph{combinable} with another order pair $({\geqslant_2},{>_2})$ on
terms  if $({\geqslant_1},{>_1})$ is normal and every argument position of any
function symbol is $>_1$-monotone or $\geqslant_2$-invariant.
If the order of combination can be inferred from the context,
we may simply call them combinable.
\end{definition}

If we consider the equivalence relation $\sim$ induced by $\geqslant$,
invariance is equivalent to
\(
f(t_1,\ldots,t_i,\ldots,t_n) \sim f(t_1,\ldots,u,\ldots,t_n)
\)
for all $\seq[n]{t}, u$. This justifies its name.

Recall the form of linear polynomial interpretations $\AA$:
\[
f_\AA(\seq{x}) = a_0 + a_1x_1 + \cdots + a_nx_n
\]
Trivially, the $i$-th argument position of $f$ is monotone with respect to $>_\AA$ if $a_i > 0$,
and invariant with respect to $\geqslant_\AA$ otherwise.
Moreover, $({\geqslant_\AA},{>_\AA})$ is normal.

\begin{example}[continued from \cref{ex:lex-comb-redpair-invalid}]
The first argument position of $\m{f}$ is neither monotone with respect to
$>_\AA$ nor invariant with respect to $\geqslant_\BB$.  Thus,
$({\geqslant_\AA},{>_\AA})$ is \emph{not} combinable with
$({\geqslant_\BB},{>_\BB})$.
\end{example}

We have similar facts for matrix interpretations $\AA$:
let the interpretation of an $n$-ary function symbol $f$ be
\(
f_\AA(\seq{\vec{x}}) = A_1\vec{x}_1 + \cdots + A_n\vec{x}_n + \vec{a}
\).
The $i$-th argument position of $f$ is monotone if $(A_i)_{1, 1} > 0$,
and invariant if $A_i = O$, see~\cite{EWZ08}.
Again, $({\geqslant_\AA},{>_\AA})$ is normal.

\begin{lemma}
\label{lem:contexts}
Let  $({\geqslant_1}, {>_1})$ be a reduction pair combinable with 
another reduction pair $({\geqslant_2}, {>_2})$.
Then $\geqslant_{12}$ is closed under contexts.
\end{lemma}
\begin{proof}
It is sufficient to show the monotonicity of $\geqslant_{12}$, meaning that
if $t_i \geqslant_{12} u$ then $C[t_i] \geqslant_{12} C[u]$ for an arbitrary 
context $C$ of the form $f(t_1,\ldots,\square,\ldots,t_n)$.  
Suppose $t_i \geqslant_{12} u$.  If it follows from
$t_i \geqslant_1 u$ and $t_i \geqslant_2 u$,
then $C[t_i] \geqslant_{12} C[u]$ follows from the assumptions
that $\geqslant_1$ and $\geqslant_2$ are closed under contexts.
Otherwise, $t_i \geqslant_{12} u$ follows from $t_i >_1 u$.
According to the combinability, the $i$-th argument position of $f$ is 
monotone in $>_1$ or invariant in $\geqslant_2$.  In the former case 
$C[t_i] >_1 C[u]$ follows from $t_i >_1 u$.
In the latter case,
$t_i \geqslant_1 u$ from the normality, and therefore  
$C[t_i] \geqslant_1 C[u] $
from the closure under contexts.
The invariance yields
$C[t_i] \geqslant_2 C[u]$.  Hence, in either case 
$C[t_i] \geqslant_{12} C[u]$ is concluded.
\qed
\end{proof}

\begin{theorem}
\label{thm:combinability}
Lexicographic combinations of combinable reduction pairs are reduction pairs.
\end{theorem}

\begin{remark}
\label{rem:normality}
The normality requirement cannot be dropped from the combinability
criterion (\cref{thm:combinability}), see~\cref{sec:omitted}.
However, normality is a mild requirement, in the sense that
typical reduction pairs like the Knuth--Bendix order~\cite{KB70},
the recursive path order~\cite{D82,KL80},
the matrix interpretation~\cite{EWZ08}, and
the polynomial interpretation~\cite{L79} are all normal reduction pairs.
Furthermore, lexicographic combinations of normal reduction pairs are normal.
\end{remark}

We demonstrate a termination proof based on lexicographic combination.

\begin{example}
\label{ex:plus1}
Recall the TRS $\RR$ in the introduction:
\begin{align*}
1\colon\:
\m{s}(x) + y & \to \m{p}(\m{s}(x)) + \m{s}(y)
&
2\colon\: \m{p}(\m{s}(x)) & \to x
\end{align*}
The set $\DP(\RR)$ consists of the two dependency pairs:
\begin{align*}
3\colon\:
\m{s}(x) +^\sharp y & \to \m{p}(\m{s}(x)) +^\sharp \m{s}(y)
&
4\colon\:
\m{s}(x) +^\sharp y & \to \m{p}^\sharp(\m{s}(x))
\end{align*}
Let $\AA$ and $\BB$ be the linear interpretations given by:
\begin{align*}
\m{s}_\AA(\underline{x}) &= x + 1
&
\m{p}_\AA(\underline{x}) &= x
&
\m{p}^\sharp_\AA(x) &= 0
&
\underline{x} +_\AA y &= x
&
\underline{x} +^\sharp_\AA y &= x
\\
\m{s}_\BB(x) &= x + 1
&
\m{p}_\BB(\overline{x}) &= 0
&
\m{p}^\sharp_\BB(\overline{x}) &= 0
&
x +_\BB \overline{y} &= x
& 
x +^\sharp_\BB \overline{y} &= x
\end{align*}
Here monotone positions in $\AA$ and invariant positions in $\BB$ are
indicated by underlining and overlining, respectively.  The induced
reduction pairs are combinable, so let
$({\geqslant_{\AA\BB}},{>_{\AA\BB}})$ be their lexicographic combination.
Then $\geqslant_{\AA\BB}$ orients the rules in $\RR$ while $>_{\AA\BB}$
orients those in $\DP(\RR)$:
\begin{align*}
1\colon\: & x+1 \geqslant_\AA x+1
&
2\colon\: & x+1 >_\AA x
&
3\colon\: & x+1 \geqslant_\AA x+1
&
4\colon\: & x+1 >_\AA 0
\\
1\colon\:
& x+1 \geqslant_\BB 0
&
&
&
3\colon\:
& x+1 >_\BB 0
\end{align*}
Hence, the termination of $\RR$ is concluded by \cref{thm:dp,thm:rp}.
\end{example}

\begin{remark}
\label{rem:rp}
The following fact gives an alternative method for combining reduction
pairs lexicographically: Given reduction pairs $({\geqslant_1},{>_1})$ and
$({\geqslant_2},{>_2})$, the pair
$({\geqslant_1} \cap {\geqslant_2},{>_{12}})$ forms a reduction pair, where
$>_{12}$ is defined as in \cref{def:lex-comb}.
\cref{thm:rp} with such a reduction pair is equivalent to successive
application of the same theorem with the two.  However, this approach cannot handle
\cref{ex:plus1} with the reduction pairs used there, as it imposes the
additional constraint $\m{p}(\m{s}(x)) \geqslant_\BB x$.
This prevents a successful termination proof, as mentioned in the
introduction.  It is worth noting that the use of the usable rule
criterion~\cite{HM07,TGS04} for \cref{thm:rp} is not helpful here
as $\m{p}(\m{s}(x)) \to x$ is a usable rule.
\end{remark}

If monotone and invariant positions of combined reduction pairs are
identified, one can combine multiple reduction pairs by successively applying
\cref{thm:combinability}.
Let $({\geqslant_{12}}, {>_{12}})$ be the lexicographic combination of
combinable reduction pairs $({\geqslant_1}, {>_1})$ and 
$({\geqslant_2}, {>_2})$.

\begin{lemma}
\label{lem:propagation}
If the $i$-th argument position of a function symbol is $>_2$-monotone
and $>_2$ is non-empty then the position is $>_1$-monotone.  Similarly, if
the $i$-th argument position of a function symbol is
$\geqslant_1$-invariant and $>_1$ is non-empty then the position is
$\geqslant_2$-invariant.
\end{lemma}
\begin{proof}
We only show the first claim because the second is shown in a similar way.
Due to the non-emptiness of $>_2$, there are terms $t$ and $u$ with 
$t >_2 u$. By monotonicity
\(
f(t,\ldots, t,\ldots, t) >_2
f(t, \ldots, u, \ldots, t)
\)
is obtained.  
If the $i$-th argument position of $f$ were $\geqslant_2$-invariant, we
would also have 
\(
f(t,\ldots, u, \ldots, t) \geqslant_2 
f(t,\ldots, t, \ldots, t)
\),
which leads to a contradiction.
So the $i$-th argument position of $f$ is not
$\geqslant_2$-invariant, and therefore the claim follows from the
combinability.
\qed
\end{proof}

\begin{theorem}
\label{thm:propagation}
If the $i$-th argument of a function symbol is $>_2$-monotone and ${>_2}$ is non-empty,
then it is $>_{12}$-monotone too. Similarly,
if the $i$-th argument of a function symbol is $\geqslant_1$-invariant and
$>_1$ is non-empty, then it is $\geqslant_{12}$-invariant too. 
\end{theorem}
\begin{proof}
Again, we only show the first claim.
Let
$t = f(t_1, \ldots, t_i, \ldots, t_n)$, $u = f(t_1,\ldots,u',\ldots,t_n)$,
and $t_i >_{12} u'$.  If $t_i >_{12} u'$ is due to $t_i >_1 u'$ then
\cref{lem:propagation} yields $t >_1 u$, from which $t >_{12} u$ follows.
Otherwise, $t_i \geqslant_1 u'$ and $t_i >_2 u'$.  As $\geqslant_1$ is
closed under contexts and the $i$-th position is $>_2$-monotone, 
the inequalities $t \geqslant_1 u$ and $t >_2 u$ follow.
Hence, $t >_{12} u$ is concluded.
\qed
\end{proof}

The non-emptiness condition in \cref{thm:propagation} cannot be dropped,
see~\cref{sec:omitted}.
In practice,
this is not a problem, because a usual reduction pair
(including those mentioned in \cref{rem:normality})
can be extended to another one $({\geqslant}, {>})$
so that satisfies $\m{c}_1 > \m{c}_2$ for fresh constant symbols $\m{c}_1$ and $\m{c}_2$.

\begin{example}[continued from \cref{ex:plus1}]
\label{ex:plus2}
We show the termination of the extended system $\RR' = \RR \cup \{ \m{0} + y \to y \}$.
The set $\DP(\RR')$ coincides with $\DP(\RR)$.  The linear interpretation $\CC$ with
\begin{align*}
\m{0}_\CC &= 1
&
\m{s}_\CC(\underline{x}) &= x 
&
\m{p}_\CC(\underline{x}) &= x
&
\m{p}^\sharp_\CC(x) &= 0
&
\underline{x} +_\CC \underline{y} &= x + y
&
\underline{x} +^\sharp_\CC y &= x
\end{align*}
satisfies $\m{0} + y >_\CC y$ and 
$\RR' \cup \DP(\RR) \subseteq {\geqslant_\CC}$.
Therefore, $\RR \subseteq {\geqslant_{\CC\AA\BB}}$ and
$\DP(\RR) \subseteq {>_{\CC\AA\BB}}$ are obtained
if we extend $\AA$ and $\BB$ with $\m{0}_{\AA} = 0$ and $\m{0}_{\BB} = 0$.
The combinability 
can be verified by \cref{thm:propagation}. 
Hence, $\RR$ is terminating.
\end{example}

As a side note, order of combination does not matter to combinability:
Let $({\geqslant_1},{>_1})$, $({\geqslant_2},{>_2})$, and
$({\geqslant_3},{>_3})$ be reduction pairs such that
$>_1$, $>_2$, and $>_3$ are non-empty. 
If $({\geqslant_1},{>_1})$ is combinable with $({\geqslant_2},{>_2})$ and also
$({\geqslant_2},{>_2})$ with $({\geqslant_3},{>_3})$,
then, by  \cref{thm:propagation},  $({\geqslant_{12}},{>_{12}})$ is combinable with
$({\geqslant_3},{>_3})$ and 
$({\geqslant_1},{>_1})$ with $({\geqslant_{23}},{>_{23}})$.

\section{Termination of the Battle of Hercules and Hydra}
\label{sec:hydra}

In this short section, we demonstrate a termination proof based on
heterogeneous combination of reduction pairs.  Among others,
we pick up Touzet's TRS encoding of the
Battle of Hercules and Hydra~\cite{T98}; see also~\cite{DM07}
for its backgrounds. Touzet's rewrite system $\HH$ consists of the
following eleven rules:
\begin{alignat*}{4}
1\colon &\;&
{\talloblong\,\circ}~x &\to {\circ\,\talloblong}\,x
&\hspace{4em}
6\colon &\;&
\m{H}(\m{0},x) &\to {\circ}~x
\\
2\colon &&
{\bullet\,\talloblong}\,x &\to {\talloblong\bullet\bullet}~x
&
7\colon &&
\bullet~\m{H}(\m{H}(\m{0},y),z) &\to \m{c^1}(y,z)
\\
3\colon &&
{\circ}~x &\to {\bullet\,\talloblong}\,x
&
8\colon &&
{\bullet}~\m{H}(\m{H}(\m{H}(\m{0},x),y),z) &\to
\m{c^2}(x,y,z)
\\
4\colon &&
{\bullet}~x &\to x
&
9\colon &&
{\bullet}~\m{c^1}(x,y) &\to \m{c^1}(x,\m{H}(x,y))
\\
5\colon &&
\m{c^1}(y,z) &\to {\circ}~z
&
10\colon &&
{\bullet}~\m{c^2}(x,y,z) &\to \m{c^2}(x,\m{H}(x,y),z)
\\
&&&
&
11\colon &&
\m{c^2}(x,y,z) &\to {\circ}~\m{H}(y,z)
\end{alignat*}
Symbols $\talloblong$, $\circ$, and $\bullet$ are unary function symbols
and their parentheses are omitted here.  

While the original termination proof of $\HH$ uses a sophisticated algebra on
$\mathbb{O} \times \NN \times \NN$, our proof employs the combination of an
ordinal interpretation on $\mathbb{O}$ and a Knuth--Bendix order.  Here
$\mathbb{O}$ is the set of all ordinal numbers below $\varepsilon_0$.  
In order to ease the proof, we introduce an easy corollary of
\cref{thm:dp,thm:rp}.
Below, we write $\Emb$ for the TRS consisting of the \emph{embedding rule}
$f(\seq{x}) \to x_i$ for all $n$-ary function symbols $f$ and $1 \leqslant i
\leqslant n$.

\begin{corollary}
\label{cor:rp}
A TRS $\RR$ is terminating if
$\RR \subseteq {>}$ and ${\Emb} \subseteq {\geqslant}$ for some
reduction pair $({\geqslant},{>})$.
\end{corollary}
\begin{proof}
Define $s \geqslant^\natural t$ and $s >^\natural t$ as 
$s^\natural \geqslant t^\natural$ and $s^\natural > t^\natural$,
respectively.  Here $t^\natural$ replaces all marked symbols $f^\sharp$ in
$t$ by the corresponding unmarked symbols $f$. Then
$({\geqslant^\natural},{>^\natural})$ is a reduction pair with 
$\RR \subseteq {\geqslant^\natural}$.
To show termination by \cref{thm:dp,thm:rp}, it remains to show 
$\DP(\RR) \subseteq {>^\natural}$.
Since every dependency pair $\ell^\sharp \to t^\sharp$ admits a rule
$\ell \to r \in \RR$ with $r \to_\Emb^* t$, we have $\ell > r \geqslant t$ and thus
$\ell^\sharp >^\natural t^\sharp$.
\qed
\end{proof}

Actually, we use the Knuth--Bendix order together with argument filtering.
For verifying combinability, we need to identify monotone and invariant
positions.  Consider an arbitrary argument filter $\pi$.  For brevity, we
write $i \in \pi(f)$ if $\pi(f) = i$ or $\pi(f) = [\ldots,i,\ldots]$. Let
$({\geqslant}, {>})$ be a monotone reduction pair.  Not surprisingly, the
$i$-th argument position of a function symbol $f$ is $>^\pi$-monotone if
$i \in \pi(f)$, and $\geqslant^\pi$-invariant otherwise.

To prove the termination of $\HH$, we use the same algebra $\OO$
on $\mathbb{O}$ as Touzet's original termination proof, that is:
\begin{align*}
\m{0}_\OO &= 0
&
\m{c}^\m{1}_\OO(\underline{x},y) &= y + \omega^{x + 1}
&
\circ_\OO(\underline{x}) &= x
&
{\talloblong_\OO}(\underline{x}) &= x
\\
\m{H}_\OO(\underline{x},\underline{y})
&= \omega^x \oplus y
&
\m{c}^\m{2}_\OO(\underline{x},y,\underline{z})
&= z \oplus \omega^{y + \omega^{x + 1}}
&
\bullet_\OO(\underline{x}) &= x
\end{align*}
Here $\oplus$ stands for natural addition
and monotone positions are
indicated by underlining.  Remark that the second argument positions of
$\m{c}^1$ and $\m{c}^2$ are not monotone, since the first argument position
of $+$ on ordinals is not monotone (e.g. $1 + \omega = 0 + \omega$).
According to~\cite[Lemma~3]{T98}, we have
\(
\{5,6,11\} \cup (\Emb \setminus \{12\,\text{--}\,14\})
\subseteq {>_\OO}
\)
and
\(
\{1\,\text{--}\,4,7\,\text{--}\,10,12\,\text{--}\,14\}
\subseteq {\geqslant_\OO}
\),
where $12\,\text{--}\,14$ denote the rules ${\circ}\,x \to x$,
${\bullet}\,x \to x$, and ${\talloblong}\,x \to x$ in $\Emb$,
respectively.  By taking the argument filter $\pi$ with
$\pi(\m{0}) = \pi(\m{H}) = \pi(\m{c^1}) = \pi(\m{c^2}) = \emptylist$,
and $\pi(\circ) = \pi(\bullet) = \pi(\talloblong) = [1]$
the rules in the latter group are simplified as follows:
\begin{alignat*}{8}
1\colon &\;&
{\talloblong\,\circ}~x &\to {\circ\,\talloblong}~x
&\qquad
4\colon &\;&
{\bullet}~x &\to x
&\qquad
9\colon &\;&
{\bullet}~\m{c^1} &\to \m{c^1}
&\qquad
12\colon &\;&
{\circ}~x &\to x
\\
2\colon &&
{\bullet\,\talloblong}~x &\to {\talloblong\bullet\bullet}~x
&
7\colon &\;&
\bullet~\m{H} &\to \m{c^1}
&
10\colon &&
{\bullet}~\m{c^2} &\to \m{c^2}
&
13\colon &&
{\bullet}~x &\to x
\\
3\colon &\;&
{\circ}~x &\to {\bullet\,\talloblong}~x
&
8\colon &&
{\bullet}~\m{H} &\to
\m{c^2}
&
&&
&
&
14\colon &&
{\talloblong}~x &\to x
\end{alignat*}
Consider the Knuth--Bendix order $>_\kbo$ induced from the precedence 
${\bullet} \succ {\talloblong} \succ {\circ} \succ \m{c^1} \succ \m{c^2}$
and the admissible weight function $(w_0,w)$ with:
\begin{align*}
w({\circ}) &= 2
&
w_0 &= w(\m{0}) = w(\m{H}) = w(\m{c^1}) = w(\m{c^2}) = w({\talloblong}) = 1
&
w({\bullet}) &= 0
&
\end{align*}
It is not difficult to see that the Knuth--Bendix order orients the above simplified
rules strictly,
and thus
\(
\{1\,\text{--}\,4,7\,\text{--}\,10,12\,\text{--}\,14\}
\subseteq {>_\kbo^\pi}
\).
As non-monotone argument positions with respect to $>_\OO$ are filtered out
by $\pi$, the reduction pairs $({\geqslant_\OO},{>_\OO})$ and
$({\geqslant_\kbo^\pi},{>_\kbo^\pi})$ are combinable.
Hence, the termination is concluded by \cref{cor:rp}.

\begin{remark}
\label{rem:rule-removal}
We stress that the use of \cref{thm:combinability} is crucial for obtaining
a termination proof (see \cref{sec:hydra-analysis}).  In particular, the
rule removal method~\cite{TGS04} is not applicable because the reduction
pair induced from Touzet's interpretation $\OO$ lacks monotonicity.
\end{remark}

\section{Echelon-Form Matrix Interpretation}
\label{sec:matrix}

In this section, we investigate another technique to construct reduction
pairs based on lexicographic comparison.
The underlying observation is that, linear polynomial interpretations
combined lexicographically (by the combinability criterion) can be regarded
as a variant of the matrix interpretation which uses the lexicographic
order pairs $({\geqslant^\lex}, {>^\lex})$.
Here $(\seq[n]{x})^T >^\lex (\seq[n]{y})^T$ if there exists an index
$1 \leqslant i \leqslant n$ such that $x_i > y_i$ and $x_j = y_j$ for all
$j < i$. The relation $\geqslant^\lex$ is the reflexive closure of $>^\lex$.

\begin{example}[continued from \cref{ex:plus2}]
\label{ex:plus3}
The three linear interpretations $\CC$, $\AA$, and $\BB$ can be
combined into the single algebra $\MM$ on $\NN^3$ with the interpretations:
\begin{align*}
\m{s}_\MM(\vec{x}) &= 
\begin{pmatrix}
1 & 0 & 0 \\
0 & 1 & 0 \\
0 & 0 & 1 \\
\end{pmatrix} \vec{x} +
\begin{pmatrix}
    0 \\
    1 \\
    1
\end{pmatrix}
&
\m{p}_\MM(\vec{x}) &=
\begin{pmatrix}
1 & 0 & 0 \\
0 & 1 & 0 \\
0 & 0 & 0 \\
\end{pmatrix} \vec{x}
&
\m{0}_\MM &= \vec{e}_1
\\
\vec{x} +_\MM \vec{y} &= 
\begin{pmatrix}
1 & 0 & 0 \\
0 & 1 & 0 \\
0 & 0 & 1 \\
\end{pmatrix} \vec{x} +
\begin{pmatrix}
1 & 0 & 0 \\
0 & 0 & 0 \\
0 & 0 & 0 \\
\end{pmatrix} \vec{y}
&
\vec{x} +^\sharp_\MM \vec{y} &= 
\begin{pmatrix}
1 & 0 & 0 \\
0 & 1 & 0 \\
0 & 0 & 1 \\
\end{pmatrix} \vec{x}
&
\m{p}_\MM^\sharp(\vec{x}) &= \vec{0}
\end{align*}
For instance, 
\(
\m{p}_\MM(\vec{x}) =
(x_1,x_2,0) =
(\m{p}_\CC(x_1),\m{p}_\AA(x_2),\m{p}_\BB(x_3))
\)
holds for $\vec{x} = (x_1,x_2,x_3)^T$.  If we 
equip $\MM$ with $\geqslant^\lex$,
the reduction pair $({\geqslant_\MM},{>_\MM})$ coincides with the
lexicographic combination employed in \cref{ex:plus2}.
\end{example}

Let $\MM$ be such a matrix interpretation equipped with the lexicographic
order.  Although $\MM$ is well-founded, it cannot afford weak monotonicity
for free (cf.~\cref{ex:lex-comb-redpair-invalid}),
which is a relevant property for $({\geqslant_\MM}, {>_\MM})$ to be a
reduction pair.  Below, we show that weak monotonicity is characterized by
\emph{(column) echelon-form} matrices.

Let $A$ be an $m \times n$ matrix.
The matrix $A$ is in \emph{(column) echelon form} if 
the following property holds for all $1 \leqslant i \leqslant m$ and 
$2 \leqslant j \leqslant n$:
If $A_{k, \, j-1} = 0$ for all $k < i$ then $A_{i, j} = 0$.
By definition, every echelon-form matrix is a lower-triangular matrix.  In
particular, $A_{1, j} = 0$ is enforced for every $2 \leqslant j$, as the
premise is vacuously satisfied.

\begin{example}
\label{ex:echelon}
The $2 \times 2$ and $3 \times 3$ echelon-form matrices are classified as follows:
\[
\begin{pmatrix} 0 & 0 \\ \mathord{*} & 0 \end{pmatrix}
\quad
\begin{pmatrix} \mathord{+} & 0 \\ \mathord{*} & \mathord{*} \end{pmatrix}
\quad
\begin{pmatrix}
    0 & 0  & 0 \\
    0 & 0  & 0 \\
    \mathord{*} & 0  & 0
\end{pmatrix}
\quad
\begin{pmatrix}
    0 & 0  & 0 \\
    \mathord{+} & 0  & 0 \\
    \mathord{*} & \mathord{*}  & 0
\end{pmatrix}
\quad
\begin{pmatrix}
    \mathord{+} & 0  & 0 \\
    \mathord{*} & 0  & 0 \\
    \mathord{*} & \mathord{*} & 0
\end{pmatrix}
\quad
\begin{pmatrix}
    \mathord{+} & 0  & 0 \\
    \mathord{*} & \mathord{+}  & 0 \\
    \mathord{*} & \mathord{*} & \mathord{*}
\end{pmatrix}
\]
Here $\mathord{+}$ stands for an arbitrary positive natural number, and 
$\mathord{*}$ for an arbitrary non-negative integer.
The matrices
$\begin{psmallmatrix} 0 & 0 \\ 0 & 1 \end{psmallmatrix}$ and 
$\begin{psmallmatrix}
    0 & 0  & 0 \\
    0 & 0  & 0 \\
    0 & 1  & 0
\end{psmallmatrix}$
are not in echelon form, so not all lower-triangular matrices are in
echelon form.
\end{example}

\begin{lemma}
\label{lem:echelon-weak-mono}
Let $A$ be an $m \times n$ matrix in echelon form.
Then $A \vec{x} \geqslant^\lex A \vec{y}$ whenever $\vec{x} \geqslant^\lex \vec{y}$.
\end{lemma}
\begin{proof}
We proceed by mathematical induction on $m$.
If $n = 1$ then the claim is trivial, so assume $n > 1$
and let 
\(
\vec{x} = 
\begin{psmallmatrix} x_1 \\ \vec{x'} \end{psmallmatrix}
\)
and
\(
\vec{y} = 
\begin{psmallmatrix} y_1 \\ \vec{y'} \end{psmallmatrix}
\).
We further analyze the first column of $A$.
If $A_{i,1} = 0$ for all $1 \leqslant i < m$ (which is in particular true when $m = 1$),
from the echelon formedness it holds that
$A_{i, j} = 0$ except for $i = m$ and $j = 1$, and therefore
\[
A \vec{x} =
\begin{pmatrix}
\vec{0}\\
A_{m,1} x_1
\end{pmatrix}
\geqslant^\lex
\begin{pmatrix}
\vec{0}\\
A_{m,1} y_1
\end{pmatrix}
= A\vec{y}.
\]
Otherwise, 
by calculation, we have:
\begin{align*}
A & =
\begin{pmatrix}
    \vec{0} & & O         & \\
    A_{i,1} & & \vec{0}^T & \\
    \vec{a} & & A'        & \\
\end{pmatrix}
&
A \vec{x} & =
\begin{pmatrix}
\vec{0}\\
A_{i, 1} x_1\\
\vec{a} x_1 + A' \vec{x'} \\
\end{pmatrix}
&
A \vec{y} & =
\begin{pmatrix}
\vec{0}\\
A_{i, 1} y_1\\
\vec{a} y_1 + A' \vec{y'} \\
\end{pmatrix}
\end{align*}
Here $A_{i, 1}$ is the first positive entry at the $i$-th row with 
$1 \leqslant i < m$, 
$\vec{a}$ is a vector of length $m - i > 0$, and $A'$ is an $(m-i) \times (n-1)$ echelon-form matrix.
If $\vec{x} \geqslant^\lex \vec{y}$ is due to $x_1 > y_1$
then $A_{i,1} x_1 > A_{i, 1} y_1$ and therefore $A \vec{x} \geqslant^\lex A \vec{y}$.
Otherwise, $x_1 = y_1$ and $\vec{x'} \geqslant^\lex \vec{y'}$.
By the induction hypothesis $A' \vec{x'} \geqslant^\lex A' \vec{y'}$
and therefore $A \vec{x} \geqslant^\lex A \vec{y}$. \qed
\end{proof}

\begin{lemma}
\label{lem:weak-mono-echelon}
Let $A$ be an $m \times n$ matrix.  If $\vec{x} \geqslant^\lex \vec{y}$
implies $A \vec{x} \geqslant^\lex A \vec{y}$ for all vectors $\vec{x},
\vec{y}$ of natural numbers, then $A$ is in echelon form.
\end{lemma}

See \cref{sec:omitted} for the proof of \cref{lem:weak-mono-echelon}.
Let $\MM$ be an algebra whose carrier is the set of vectors of natural
numbers ordered lexicographically
and interpretations $f_\MM(\vec{x}_1, \ldots, \vec{x}_n) = A_1 \vec{x}_1 + \cdots + A_n \vec{x}_n + \vec{a}$
are built from echelon-form matrices $A_1, \ldots, A_n$.
We dub such an algebra \emph{echelon-form matrix interpretation}.

\begin{theorem}
\label{thm:echelon-redpair}
The pair $({\geqslant_\MM}, {>_\MM})$ is a reduction pair
for every echelon-form matrix interpretation $\MM$.
\end{theorem}

As expected, a lexicographic combination of $d$ linear polynomial
interpretations (satisfying the combinability criterion) corresponds to a
$d$-dimensional echelon-form matrix interpretation whose non-diagonal
entries are all zero.

\begin{theorem}
\label{thm:lex-pol-echelon}
Let $d$ be a positive integer, and let $\AA_1, \ldots, \AA_d$ be linear
polynomial interpretations such that
$\AA_i$ and $\AA_{i+1}$ are combinable for all $1 \leqslant i < d$.
Let $({\geqslant}, {>})$ is the lexicographic combination of $\AA_1, \ldots, \AA_d$.
Then there is a $d$-dimensional echelon-form matrix interpretation $\MM$ such that
$({\geqslant, {>}})$ is identical to $({\geqslant_\MM}, {>_\MM})$.
\end{theorem}

To characterize strict monotonicity for $m \times n$ matrices $A$
(i.e., the property that $\vec{x} >^\lex \vec{y}$ implies $A \vec{x} >^\lex A \vec{y}$),
it suffices to additionally assume $n \leqslant m$ and
$A_{i, i} > 0$ for all $1 \leqslant i \leqslant n$.
Such an echelon-form matrix is called \emph{positive}.

\begin{example}
\label{ex:positive-echelon}
The echelon-form matrix 
$A = \begin{psmallmatrix} 1 & 0 \end{psmallmatrix}$
is not positive, while $A^T$ is.
Indeed, $A$ is not strictly monotone, as witnessed by
\(
A \begin{psmallmatrix}
0 \\ 1   
\end{psmallmatrix}
= \begin{psmallmatrix}
0 \\ 0
\end{psmallmatrix}
=
A \begin{psmallmatrix}
0 \\ 0
\end{psmallmatrix}
\).
\end{example}

\begin{lemma}
\label{lem:positive-echelon-strict-mono}
Let $A$ be an $m \times n$ positive echelon-form matrix.
Then $A \vec{x} >^\lex A \vec{y}$ whenever $\vec{x} >^\lex \vec{y}$.
\end{lemma}
\begin{proof}
We proceed by mathematical induction on $n$.
The case $n = 1$ is trivial.
Otherwise, write $\vec{x}, \vec{y}$ and $A$ as follows:
\begin{align*}
\vec{x} &= \begin{pmatrix}
x_1 \\
\vec{x'}
\end{pmatrix}
&
\vec{y} &= \begin{pmatrix}
y_1 \\
\vec{y'}
\end{pmatrix}
&
A & =
\begin{pmatrix}
    A_{1,1} &       & \vec{0}^T &  \\
    \vec{a} &       & A'        &   \\
\end{pmatrix}
\end{align*}
Here $A_{1, 1}$ is the first positive entry, 
$\vec{a}$ is a vector of length $m - 1$,
and $A'$ is an $(m-1) \times (n-1)$ positive echelon-form matrix.
A calculation shows that
\begin{align*}
A \vec{x} & =
\begin{pmatrix}
A_{1, 1} x_1\\
\vec{a} x_1 + A' \vec{x'} \\
\end{pmatrix}
&
A \vec{y} & =
\begin{pmatrix}
A_{1, 1} y_1\\
\vec{a} y_1 + A' \vec{y'} \\
\end{pmatrix}
\end{align*}
If $\vec{x} >^\lex \vec{y}$ is by $x_1 > y_1$
then $A_{1,1} x_1 > A_{1, 1} y_1$ and therefore $A \vec{x} >^\lex A \vec{y}$.
Otherwise, $x_1 = y_1$ and $\vec{x'} >^\lex \vec{y'}$.
By the induction hypothesis $A' \vec{x'} >^\lex A' \vec{y'}$
and therefore $A \vec{x} >^\lex A \vec{y}$. \qed
\end{proof}

\begin{lemma}
\label{lem:strict-mono-positive-echelon}
Let $A$ be an $m \times n$ matrix.  If $\vec{x} >^\lex \vec{y}$
implies $A \vec{x} >^\lex A \vec{y}$ for all vectors $\vec{x},
\vec{y}$ of natural numbers, then $A$ is a positive echelon-form matrix.
\end{lemma}

See \cref{sec:omitted} for the proof of \cref{lem:strict-mono-positive-echelon}.
Now we can characterize monotone and invariant positions of echelon-form matrix
interpretations $\MM$. Let the interpretation of an $n$-ary function symbol $f$
be
\(
f_\MM(\seq{\vec{x}}) = A_1\vec{x}_1 + \cdots + A_n\vec{x}_n + \vec{a}
\).
The $i$-th argument position of $f$ is monotone if $A_i$ is positive,
and invariant if $A_i = O$.
Besides, $({\geqslant_\MM},{>_\MM})$ is normal.

We again tame the Hydra, using an echelon-form matrix interpretation which
has a positive non-diagonal entry and therefore goes beyond a combination
of linear polynomial interpretations.

\begin{example}
\label{ex:hydra-echelon}
Recall the TRS $\HH$ in \cref{sec:hydra}. Instead of the Knuth--Bendix
order, we establish the termination by the combination of the ordinal
interpretation $\OO$ with the following echelon-form matrix interpretation
$\MM$:
\begin{align*}
{\bullet}_\MM(\vec{x}) &=
\begin{pmatrix}
1 & 0 \\
1 & 1 
\end{pmatrix}\!\vec{x} +
\begin{pmatrix}
0 \\ 1
\end{pmatrix}
&
{\circ}_\MM(\vec{x}) &=
\begin{pmatrix}
2 & 0 \\
0 & 1 
\end{pmatrix}\!\vec{x} +
\begin{pmatrix}
3 \\ 0
\end{pmatrix}
&
{\talloblong}_\MM(\vec{x}) &=
\begin{pmatrix}
2 & 0 \\
0 & 1 
\end{pmatrix}\!\vec{x} +
\begin{pmatrix}
2 \\ 0
\end{pmatrix}
\end{align*}
The other interpretations $\m{0}_\MM$,
$\m{H}_\MM(\overline{\vec{x}}, \overline{\vec{y}})$,
$\m{c^1}_\MM(\overline{\vec{x}}, \overline{\vec{y}})$, and
$\m{c^2}_\MM(\overline{\vec{x}}, \overline{\vec{y}}, \overline{\vec{z}})$
are defined as the constant vector $\vec{0}$.
One can confirm the combinability,
$\RR \subseteq {>_{\OO\MM}}$, and
$\Emb \subseteq {\geqslant_{\OO\MM}}$.
For instance, the orientations
${\talloblong\,\circ}~x >_\MM {\circ\,\talloblong}~x$ and
${\bullet\,\talloblong}~x >_\MM {\talloblong\,\bullet\bullet}~x$
of the first two rules in $\HH$ are verified as follows:
\begin{align*}
1\colon
\begin{pmatrix}
4 & 0 \\
0 & 1 
\end{pmatrix}\!\vec{x} +
\begin{pmatrix}
8 \\ 0
\end{pmatrix}
&>^\lex
\begin{pmatrix}
4 & 0 \\
0 & 1 
\end{pmatrix}\!\vec{x} +
\begin{pmatrix}
7 \\ 0
\end{pmatrix}
&
2\colon
\begin{pmatrix}
2 & 0 \\
2 & 1 
\end{pmatrix}\!\vec{x} +
\begin{pmatrix}
2 \\ 3
\end{pmatrix}
&>^\lex
\begin{pmatrix}
2 & 0 \\
2 & 1 
\end{pmatrix}\!\vec{x} +
\begin{pmatrix}
2 \\ 2
\end{pmatrix}
\end{align*}
So, the termination of $\HH$ is again concluded.
\end{example}

The standard matrix interpretation with the component-wise order and our
echelon-form matrix interpretation are incomparable. 
For example, the termination of the TRS $\RR$
\begin{align*}
\m{f}(x,\m{s}(y),z) & \to \m{f}(x,y,\m{s}(\m{s}(z)))
&
\m{f}(\m{s}(x),y,z) & \to \m{f}(x,z,y)
&
\m{f}(x, y, z) & \to y
\end{align*}
can be shown by the following echelon-form matrix interpretation $\AA$:
\begin{align*}
\m{f}_\AA^{(\sharp)}(\vec{x},\vec{y},\vec{z}) &= 
\vec{x}
+
\begin{pmatrix}
1 & 0 & 0 \\
0 & 0 & 0 \\
0 & 1 & 0
\end{pmatrix} \vec{y}
+
\begin{pmatrix}
1 & 0 & 0 \\
0 & 0 & 0 \\
0 & 0 & 0
\end{pmatrix} \vec{z}
+ \vec{e}_1
&
\m{s}_\AA(\vec{x}) &=
\vec{x} +
\vec{e}_2
\end{align*}
It is easy to verify
$\DP(\RR) \subseteq {>_\AA}$ and $\RR \subseteq {\geqslant_\AA}$.
On the other hand, any standard matrix interpretation cannot satisfy them.
This is seen by a complexity consideration:
Let $[n] = \m{s}^n(x)$ and $t_n = \m{f}^{\sharp}(\m{f}([n], [1], [0]), [1], [0])$.
Then 
\[
t_n
\to^{*}_\RR
\m{f}^\sharp([2^n], [1], [0])
\to^{*}_{\DP(\RR)}
\m{f}^\sharp([0], [2^{2^n}], [0])
\to^{2^{2^n}}_{\DP(\RR)}
\m{f}^\sharp([0], [0], [2^{1 + 2^n}])
\]
holds for every $n \in \NN$.
Observe that the number of rewrite steps by $\DP(\RR)$ is double exponential.
However, the interpretation of $t_n$ cannot bound such a number from above;
see the proof of~\cite[Lemma~7]{EWZ08}.

For the converse, we consider the TRS
$\RR = \{ \m{f}(\m{a}) \to \m{f}(\m{b}), \m{g}(\m{b}) \to \m{g}(\m{a}) \}$.
Then 
\(
\DP(\RR) =
\{ \m{f}^\sharp(\m{a}) \to \m{f}^\sharp(\m{b}),
\m{g}^\sharp(\m{b}) \to \m{g}^\sharp(\m{a})
\}
\).
The next standard matrix interpretation $\AA$ 
satisfies 
$\RR \subseteq {\geqslant_{\AA}}$ and 
$\DP(\RR) \subseteq {>_{\AA}}$:
\begin{align*}
\m{a}_\AA &= \begin{pmatrix}
1 \\ 0
\end{pmatrix}
&
\m{b}_\AA &= \begin{pmatrix}
0 \\ 1
\end{pmatrix}
&
\m{f}_\AA(\vec{x}) &=
\begin{pmatrix}
1 & 1\\
0 & 0
\end{pmatrix} \vec{x}
&
\m{g}_\AA(\vec{x}) &=
\begin{pmatrix}
1 & 1\\
1 & 1
\end{pmatrix} \vec{x}
\\
& & & &
\m{f}^\sharp_\AA(\vec{x}) &=
\begin{pmatrix}
1 & 1\\
0 & 1
\end{pmatrix} \vec{x}
&
\m{g}^\sharp_\AA(\vec{x}) &=
\begin{pmatrix}
0 & 1\\
0 & 0
\end{pmatrix} \vec{x}
\end{align*}
Here, it is essential that $\m{a}$ and $\m{b}$ are incomparable with
respect to the underlying component-wise order.  Indeed, there is no
echelon-form matrix interpretation $\AA$ with 
$\RR \subseteq {\geqslant_\AA}$ and $\DP(\RR) \subseteq {>_\AA}$
due to totality of the lexicographic order: 
Any interpretation $\AA$ satisfies
$\m{a} \geqslant_\AA \m{b}$ or $\m{b} \geqslant_\AA \m{a}$.
If $\m{a} \geqslant_\AA \m{b}$ holds then 
$\m{g}^\sharp(\m{a}) \geqslant_\AA \m{g}^\sharp(\m{b})$ by monotonicity.
Similarly,
if $\m{b} \geqslant_\AA \m{a}$ then
$\m{f}^\sharp(\m{b}) \geqslant_\AA \m{f}^\sharp(\m{a})$. In
either case it contradicts to $\DP(\RR) \subseteq {>_\AA}$.

\section{Experiments}
\label{sec:experiments}

In order to evaluate the presented methods, we have implemented a prototype tool
for proving termination and relative termination of (finite) TRSs.
The tool uses the dependency
pair framework (\cref{thm:dp,thm:rp}) together with two standard refinements: an
iterative cycle analysis based on strongly connected components in dependency
graphs~\cite{AG00,GAO02,HM05}, the usable rule criterion~\cite{HM07},
and the rule removal method~\cite{TGS04} by monotone reduction pairs.
For relative termination, the tool uses the relative version
of~\cref{thm:dp} (a generalization of~\cref{thm:dp-relative}), which cannot be
used with the usable rule criterion due to lack of minimality,
see~\cite{INVY17}. We compare the three classes of reduction pairs and
their lexicographic combinations.
\begin{itemize}
\item
$\mb{L}$:
lexicographic path orders~\cite{KL80} with argument filtering.
\item
$\mb{E}_d$:
echelon-form matrix interpretations on $\NN^d$ with $0,1$-matrix coefficients
equipped with the lexicographic order~(\cref{sec:matrix}).
\item
$\mb{S}_d$:
matrix interpretations on $\NN^d$ with $0,1$-matrix coefficients
equipped with the standard component-wise order~\cite{EWZ08}.
Note that $\mb{S}_1$ is the same as $\mb{E}_1$.
\end{itemize}
Suitable precedences, argument filters, and interpretations are searched by
the SMT solver Z3~\cite{MB08}; see~\cite{CGST12,ZHM09} for the
SAT/SMT encoding techniques.  The experiments were run on a
computer with Intel Core i5-1340P CPU (4.6 GHz) and 8 GB memory with $60$ seconds timeout
for each (relative) termination problem.\footnote{%
The tool and the full experimental data are available at
\url{https://www.jaist.ac.jp/project/saigawa/25cade/}.}

\begin{table}[t]
\caption{Experiments on 1528 termination problems.}
\label{tbl:standard1}
\centering
\begin{tabular}{%
@{}l@{~~}
r@{~~}r@{~~}r@{~~}r@{~~}r@{~~}
r@{~~}r@{~~}r@{~~}r@{~~}r@{~~}
r@{~~}r@{~~}r@{}}
\toprule
& $\m{L}$
& $\m{LL}$
& $\m{LLL}$
& $\m{E_1}$
& $\m{E_2}$
& $\m{E_3}$
& $\m{E_4}$
& $\m{S_2}$
& $\m{S_2S_2}$
& $\m{S_2S_2S_2}$
& $\m{E_1E_1}$
& $\m{E_1L}$
& $\m{LE_1}$
\\
proved
& 372
& 389
& 389
& 464
& 562
& 566
& 564
& 593
& 619
& 617
& 506
& 496
& 406
\\
\textit{timeout}
& \textit{8}
& \textit{8}
& \textit{8}
& \textit{8}
& \textit{14}
& \textit{28}
& \textit{84}
& \textit{24}
& \textit{39}
& \textit{60}
& \textit{8}
& \textit{8}
& \textit{8}
\\
\bottomrule
\end{tabular}
\end{table}
\cref{tbl:standard1} summarizes the experimental results on 1528 termination problems
in the TRS Standard category of the Termination Problem Database~\cite{TPDB}.
For instance, the numbers in column $\mb{E_1L}$ are read as follows: In the
aforementioned setting, lexicographic combinations of $\mb{E_1}$ (linear
interpretations) with $\mb{L}$ (the lexicographic path order with argument
filtering) proved termination of 496 TRSs, while termination analysis on 8 TRSs
did not finish within 60 seconds.
In general, combination gives us more proofs.  For example, while
$\mb{L}$ and $\mb{E_1}$ produce 503 proofs in total, the union of $\mb{L}$,
$\mb{E_1}$, $\mb{LE_1}$, and $\mb{E_1L}$ amounts to 533 proofs.
The experimental results show that the echelon-form matrix interpretation
$(\mb{E_2})$ outperforms lexicographic combination of linear polynomials
($\mb{E_1E_1}$).
The union of all methods amounts to 649 proofs, which
include eight proofs missed by the state-of-the-art termination tool
\NaTT~\cite{YKS14} (version 2.3.2).  

\begin{table}[t]
\caption{Experiments on $57$ relative termination problems.}
\label{tbl:relative1}
\centering
\begin{tabular}{%
@{}l@{~~}
r@{~~}r@{~~}r@{~~}r@{~~}r@{~~}
r@{~~}r@{~~}r@{~~}r@{~~}r@{~~}
r@{~~}r@{~~}r@{}}
\toprule
& $\m{L}$
& $\m{LL}$
& $\m{LLL}$
& $\m{E_1}$
& $\m{E_2}$
& $\m{E_3}$
& $\m{E_4}$
& $\m{S_2}$
& $\m{S_2S_2}$
& $\m{S_2S_2S_2}$
& $\m{E_1E_1}$
& $\m{E_1L}$
& $\m{LE_1}$
\\
proved
& \phantom{00}4
& \phantom{0}20
& \phantom{0}22
& \phantom{00}8
& \phantom{0}43
& \phantom{0}45
& \phantom{0}47
& \phantom{0}10
& \phantom{0}47
& \phantom{0}47
& \phantom{0}41
& \phantom{0}30
& \phantom{0}27
\\
\textit{timeout}
& \textit{0}
& \textit{0}
& \textit{0}
& \textit{0}
& \textit{0}
& \textit{0}
& \textit{0}
& \textit{0}
& \textit{0}
& \textit{0}
& \textit{0}
& \textit{0}
& \textit{0}
\\
\bottomrule
\end{tabular}
\end{table}
\cref{tbl:relative1} summarizes the experimental results on relative
termination.  The TRS Relative category in the database contains 57
relative termination problems where
\cref{thm:dp-relative} is applicable.
Two of these problems are known to be relatively non-terminating.
The union of $\mb{E_4}$ and $\mb{LL}$ amounts to 50 proofs.
We note that the union includes 5 problems missed by
the 2022 version of \NaTT, and 11 missed by
another powerful termination tool AProVE~\cite{GAB+17,KVG24}.
The following example (\texttt{INVY\_15/\#3.42}) is one of the 11 problems.\footnote{%
This problem can be solved by \NaTT.}

\begin{example}
\label{ex:relative}
Consider the relative termination problem 
of $\RR/\SS$. Here $\RR$ consists of the rules
\begin{align*}
\m{half}(\m{0}) &\to \m{0}
&
\m{lastbit}(\m{0}) &\to \m{0}
\\
\m{half}(\m{s}(\m{0})) &\to \m{0}
&
\m{lastbit}(\m{s}(\m{0})) &\to \m{s}(\m{0})
\\
\m{half}(\m{s}(\m{s}(x))) &\to \m{s}(\m{half}(x))
&
\m{lastbit}(\m{s}(\m{s}(x))) &\to \m{lastbit}(x)
\\
\m{conv}(\m{0}) &\to  \m{cons}(\m{nil},\m{0})
\\
\m{conv}(\m{s}(x)) &\to \makebox[0mm][l]{$\m{cons}(\m{conv}(\m{half}(\m{s}(x))),\m{lastbit}(\m{s}(x)))$}
\end{align*}
and $\SS = \{ \m{rand}(x) \to x,\; \m{rand}(x) \to \m{rand}(\m{s}(x)) \}$.
The set $\DP(\RR)$ consists of 
\begin{align*}
\m{half}^\sharp(\m{s}(\m{s}(x))) &\to \m{half}^\sharp(x)
&
\m{conv}^\sharp(\m{s}(x)) &\to \m{conv}^\sharp(\m{half}(\m{s}(x)))
\\
\m{lastbit}^\sharp(\m{s}(\m{s}(x))) &\to \m{lastbit}^\sharp(x)
&
\m{conv}^\sharp(\m{s}(x)) &\to \m{half}^\sharp(\m{s}(x))
\\
&&
\m{conv}^\sharp(\m{s}(x)) &\to \m{lastbit}^\sharp(\m{s}(x))
\end{align*}
The following echelon-form matrix interpretation $\AA$ on $\NN^3$ 
satisfies $\DP(\RR) \subseteq {>_\AA}$ and
$\RR \cup \SS \subseteq {\geqslant_\AA}$.
\begin{align*}
\m{0}_\AA &= \vec{0}
&
\m{nil}_\AA &= \vec{0}
&
\m{lastbit}_\AA(\vec{x}) &= \vec{e}_1
&
\m{conv}_\AA(\vec{x}) &= \vec{0}
\\
\m{s}_\AA(\vec{x}) &= \vec{x} + \vec{e}_2
&
\m{cons}_\AA(\vec{x},\vec{y}) &= \vec{0}
&
\m{lastbit}_\AA^\sharp(\vec{x}) &= \vec{x}
&
\m{conv}^\sharp_\AA(\vec{x}) &= \vec{x} + \vec{e}_1
\\
\m{half}_\AA(\vec{x}) &= 
\makebox[4em][l]{$\begin{pmatrix} 1 & 0 & 0 \\ 0 & 1 & 0 \\ 0 & 0 & 0
\end{pmatrix} \vec{x}$}
&
\m{half}^\sharp_\AA(\vec{x}) &= \vec{x}
&
\m{rand}_\AA(\vec{x}) &= 
\makebox[0mm][l]{$\begin{pmatrix} 1 & 0 & 0 \\ 0 & 0 & 0 \\ 0 & 0 & 0 
\end{pmatrix} \vec{x} + \vec{e}_1$}
\end{align*}
Since $\RR$ dominates $\SS$ and $\SS$ is non-duplicating, the above
inclusions together with \cref{thm:dp-relative} entail the termination of $\RR/\SS$.
\end{example}

\begin{remark}
\label{rem:usable-rule}
The particular usefulness of our lexicographic combination in relative termination can be
explained in terms of the usable rule criterion,
which gives fewer constraints
on the quasi-order $\geqslant$ of a reduction pair $({\geqslant}, {>})$.
As mentioned earlier in this section, the technique cannot be applied
with \cref{thm:dp-relative} for relative termination, due to absence of minimality of dependency pair problems.
In contrast, \cref{thm:combinability} can be used with \cref{thm:dp-relative} (as it does not rely on minimality)
and allows us to ignore some rules in latter components of lexicographic combination,
provided that the combinability condition is met.
\end{remark}

Further experimental comparison with the weighted path order~\cite{YKS15} and the max/plus interpretation
is found in \cref{sec:wpo-maxplus}.

\section{Conclusion}
\label{sec:conclusion}

We have presented a simple criterion for combining reduction pairs based on
monotone and invariant positions. By examples and experiments, the
criterion is shown to be complementary to existing methods of lexicographic
combination of reduction pairs.  In particular, the experiments show that
state-of-the-art tools may benefit from our method.
We have also elucidated when the matrix interpretation with the lexicographic order
induces a reduction pair.  We conclude the paper by stating related work and
future work.

Our combinability criterion (\cref{thm:combinability}) is inspired by
Touzet's work~\cite{T98}. 
It is easy to confirm the precise correspondence between
her original interpretation $\AA$ and \cref{ex:hydra-echelon}.
For instance,
${\bullet}_\AA((x,m,n)) = (x,m,m+n+1)$ corresponds to 
${\bullet}_\OO(x) = x$ and
\(
{\bullet}_\MM(\begin{psmallmatrix}m \\ n\end{psmallmatrix}) = \begin{psmallmatrix}
1 & 0 \\
1 & 1
\end{psmallmatrix} \begin{psmallmatrix}m \\ n\end{psmallmatrix} + 
\begin{psmallmatrix}
0 \\ 1
\end{psmallmatrix}
\).
We anticipate that the use of lexicographic combination eases termination
analysis of challenging rewrite systems such as Goodstein 
sequences~\cite{ZWM15}.

The ordinal interpretation in a certain form corresponds to
lexicographic combination of the linear polynomial interpretation.
To see this, consider linear polynomial interpretations
\begin{align*}
f_\AA(\seq{x}) &= \text{$a_0 + \sum_i a_i x_i$}
&
f_\BB(\seq{y}) &= \text{$b_0 + \sum_i b_i y_i$}
\end{align*}
with $a_i, b_i \in \NN$ and the interpretation $f_\OO$ on
ordinal numbers below $\omega^2$ given by:
% TODO: resolve virtical overfull here
\[
f_\OO(\omega x_1 + y_1, \ldots, \omega x_n + y_1)
= \omega (a_0 + a_1 x_1 + \cdots + a_n x_n) + (b_0 + b_1 y_1 + \cdots + b_n y_n)
\]
Here $x_i$ and $y_i$ range over $\NN$. Then, 
$({\geqslant_\OO}, {>_\OO})$ is order isomorphic to
$({\geqslant_{\AA\BB}}, {>_{\AA\BB}})$.
In general, an $n$-times combination of the linear polynomial interpretation
corresponds to an ordinal interpretation below $\omega^n$.
The same can be said for the echelon-form matrix interpretation via
\cref{thm:lex-pol-echelon}.
For instance, the echelon-form matrix interpretation $\m{rand}_\AA$ of \cref{ex:relative}
corresponds to
\(\m{rand}_\OO(\omega^2 x_1 + \omega x_2 + x_3)  = \omega^2 (x_1 + 1) \)
with $x_1, x_2, x_3 \in \NN$.  Actually, it is equivalent to
$\m{rand}_\OO(x) = x + \omega^2$ for $x < \omega^3$.

\cref{thm:lex-pol-echelon} states a correspondence between the echelon-form
matrix interpretation and a special class of lexicographic combination with
linear polynomials.  A possible line of future work is to extend this
result to a broader class of lexicographic combination.

Adapting the echelon-form matrix interpretation for AC termination is
another direction for future work.  As shown in~\cite{BL87,L79}, lexicographic
combination of \emph{non-linear} polynomial interpretations is an
effective proof method for AC termination.  We believe that the
non-linear matrix interpretation~\cite{CGO10}, using matrices
instead of vectors, is a key for the work.
Speaking of matrix, theoretical and experimental comparison of existing
matrix methods (including~\cite{KW08,NM11}) to ours is yet to be done.  In
particular, we anticipate that derivational complexity is useful to
distinguish the powers of the matrix-based methods.

\begin{credits}
\subsubsection{\ackname}
We are grateful to the anonymous reviewers for the valuable comments and
suggestions.
This research was supported by JST SPRING Grant Number JPMJSP2102 and JSPS
KAKENHI Grant Numbers JP22K11900 and JP25KJ1363.
\subsubsection{\discintname}
The authors have no competing interests to declare that are
relevant to the content of this article.
\end{credits}

\bibliographystyle{splncs04}
\bibliography{references}

\appendix
\section{Omitted Examples and Proofs}
\label{sec:omitted}

The next example shows that 
the normality requirement cannot be dropped from the combinability
criterion (\cref{thm:combinability}).

\begin{example}
\label{ex:normality}
Consider the linear polynomial interpretation $\AA$ 
defined by $\m{a}_\AA = 1$, $\m{b}_\AA = 0$, and
$\m{f}_\AA(x) = 0$. With using the identity relation $=$, the
reduction pairs $({=}, {>_\AA})$ and $({\geqslant_\AA}, {>_\AA})$
satisfy all conditions for combinability 
except normality, as witnessed by
$\m{a} >_\AA \m{b}$ and $\m{a} \neq \m{b}$.
Indeed, the preorder $\geqslant$ of the lexicographic
combination is not closed under contexts, because $\m{a} \geqslant \m{b}$
from $\m{a} >_\AA \m{b}$ but not $\m{f}(\m{a}) \geqslant \m{f}(\m{b})$, as
neither $\m{f}(\m{a}) >_\AA \m{f}(\m{b})$ nor $\m{f}(\m{a}) = \m{f}(\m{b})$
holds.
\end{example}

The next example shows that 
the non-emptiness requirement cannot be dropped from \cref{thm:propagation}.

\begin{example}
Consider the following linear polynomial interpretations $\AA$ and $\BB$:
\begin{align*}
\m{f}_\AA(x) &= 0
&
\m{a}_\AA &= 1
&
\m{b}_\AA &= 0
&\qquad
\m{f}_\BB(x) &= 0 
&
\m{a}_\BB &= 0
&
\m{b}_\BB &= 0
\end{align*}
When the signature consists of the only three symbols $\m{f}$, $\m{a}$,
and $\m{b}$, the relation $>_\BB$ is the empty relation.  Therefore, the
first argument position of $\m{f}$ is monotone with respect to $>_\BB$.  However, the
position is not monotone with respect to the strict order $>_{\AA\BB}$ of the
lexicographic combination, as $>_{\AA\BB}$ degenerates to $>_\AA$.
\end{example}

\begin{proof}[of~\cref{lem:weak-mono-echelon}]
We show that every element $A_{i, j}$ with $2 \leqslant j \leqslant n$
satisfies the following property: If $A_{k, \, j-1} = 0$ for all $k < i$
then $A_{i, j} = 0$.  We proceed by complete induction on $i$
while fixing $j$ with $2 \leqslant j \leqslant n$.  Suppose
$A_{i', j-1} = 0$ for all $i' < i$.  By the induction hypothesis
$A_{i', j} = 0$ for all $i' < i$, too. Assume to the contrary 
$A_{i, j} > 0$.  For the vectors $\vec{x} = A_{i,j}\vec{e}_{j-1}$ and
$\vec{y} = (A_{i,j-1} + 1)\vec{e}_j$ the inequality 
$\vec{x} \geqslant^\lex \vec{y}$ holds. We proceed as follows:
\[
\vec{x} \geqslant^\lex \vec{y}
\implies A\vec{x} \geqslant^\lex A\vec{y}
\implies A_{i,j-1} A_{i,j} \geqslant A_{i,j-1} A_{i,j} + A_{i,j}
\implies 0 \geqslant A_{i,j}
\]
Thus, $0 \geqslant A_{i,j} > 0$ is obtained. Contradiction.
\qed
\end{proof}

\begin{proof}[of~\cref{lem:strict-mono-positive-echelon}]
From \cref{lem:weak-mono-echelon}, we know that $A$ is in echelon form.
If $n > m$, then the right-most column of $A$ is all zero,
which implies $A \vec{e}_n = \vec{0} = A \vec{0}$, a contradiction.
So $n \leqslant m$.
Notice that from the echelon formedness,
$A_{i, i} > 0$ implies $A_{(i-1), (i-1)} > 0$ for all $2 \leqslant i \leqslant n$.
So, it suffices to see $A_{n, n} > 0$,
which is again shown by contradiction:
If $A_{n, n} = 0$, then $A \vec{e}_n = \vec{0} = A \vec{0}$. \qed
\end{proof}

We also note that an echelon-form square matrix $A$ is positive if and only
if it satisfies \emph{weak simplicity} $A \vec{x} \geqslant^\lex \vec{x}$,
% for if-direction, just feed unit vectors.
which is a relevant property for the weighted path order~\cite{YKS15}.

\section{Detailed Analysis of the Hydra Battle}
\label{sec:hydra-analysis}

We show that the termination of the TRS $\HH$ cannot be shown by
\cref{thm:dp,thm:rp} 
if we employ Touzet's ordinal interpretation $\OO$ and KBO
with argument filtering, in the following setting: We first extend $\OO$ to
the signature with marked symbols by $f^\sharp_\OO = f_\OO$, and then apply
\cref{thm:dp,thm:rp} to $(\DP(\HH), \HH)$
with the reduction pair $({\geqslant_\OO}, {>_\OO})$.
The resulting dependency pair problem is shown finite if we find a reduction pair $({\geqslant}, {>})$ satisfying:
\begin{alignat*}{4}
1\colon &\;&
{\talloblong\,\circ}~x &\geqslant {\circ\,\talloblong}\,x
&\hspace{4em}
6\colon &\;&
\m{H}(\m{0},x) &\geqslant {\circ}~x
\\
2\colon &&
{\bullet\,\talloblong}\,x &\geqslant {\talloblong\bullet\bullet}~x
&
7\colon &&
\bullet~\m{H}(\m{H}(\m{0},y),z) &\geqslant \m{c^1}(y,z)
\\
3\colon &&
{\circ}~x &\geqslant {\bullet\,\talloblong}\,x
&
8\colon &&
{\bullet}~\m{H}(\m{H}(\m{H}(\m{0},x),y),z) &\geqslant
\m{c^2}(x,y,z)
\\
4\colon &&
{\bullet}~x &\geqslant x
&
9\colon &&
{\bullet}~\m{c^1}(x,y) &\geqslant \m{c^1}(x,\m{H}(x,y))
\\
5\colon &&
\m{c^1}(y,z) &\geqslant {\circ}~z
&
10\colon &&
{\bullet}~\m{c^2}(x,y,z) &\geqslant \m{c^2}(x,\m{H}(x,y),z)
\\
&&&
&
11\colon &&
\m{c^2}(x,y,z) &\geqslant {\circ}~\m{H}(y,z)
\\
1a\colon &\;&
{\talloblong^\sharp\,\circ}~x &> {\circ^\sharp\,\talloblong}\,x
&\qquad
7a\colon &\;&
\bullet^\sharp~\m{H}(\m{H}(\m{0},y),z) &> \m{c^1}^\sharp(y,z)
\\
1b\colon &&
{\talloblong^\sharp\,\circ}~x &> {\talloblong^\sharp}\,x
&
8a\colon &&
{\bullet^\sharp}~\m{H}(\m{H}(\m{H}(\m{0},x),y),z) &>
\m{c^2}^\sharp(x,y,z)
\\
2a\colon &&
{\bullet^\sharp\,\talloblong}\,x &> {\talloblong^\sharp\bullet\bullet}~x
&
9a\colon &&
{\bullet^\sharp}~\m{c^1}(x,y) &> \m{c^1}^\sharp(x,\m{H}(x,y))
\\
2b\colon &&
{\bullet^\sharp\,\talloblong}\,x &> {\bullet^\sharp\bullet}~x
&
9b\colon &&
{\bullet^\sharp}~\m{c^1}(x,y) &> \m{H}^\sharp(x,y)
\\
2c\colon &&
{\bullet^\sharp\,\talloblong}\,x &> {\bullet^\sharp}~x
&
10a\colon &&
{\bullet^\sharp}~\m{c^2}(x,y,z) &> \m{c^2}^\sharp(x,\m{H}(x,y),z)
\\
3a\colon &&
{\circ^\sharp}~x &> {\bullet^\sharp\,\talloblong}\,x
&
10b\colon &&
{\bullet^\sharp}~\m{c^2}(x,y,z) &> \m{H}^\sharp(x,y)
\\
3b\colon &&
{\circ^\sharp}~x &> {\talloblong^\sharp}\,x
\end{alignat*}
Here, the dependency pairs from rules 5, 6 and 11 are already removed by $>_\OO$.

\begin{proposition}
\label{prop:hydra-kbo}
There are no KBO and argument filter $\pi$ satisfying the constraints
$1\,\text{--}\,11$ and $1a\,\text{--}\,10b$ (even if quasi-precedence is allowed for KBO).
\end{proposition}
\begin{proof}
Assume to the contrary that the constraints are satisfied.
Since $s \geqslant_\kbo^\pi t$ implies
$|\hat{\pi}(s)|_x \geqslant |\hat{\pi}(t)|_x$
for all $x \in \VV$, we can deduce
\begin{enumerate}[(i)]
\item
$1 \in \pi(\bullet)$ from $4$,
\item
$1 \in \pi(\talloblong^\sharp)$ and $\pi(\circ) = [1]$ from 
$1b$,
\item
$1 \in \pi(\bullet^\sharp)$ and $\pi(\talloblong) = [1]$ from 
$2c$, and
\item
$\pi(\m{c}^1) = \pi(H) = [2]$ from 
$5$, $6$, $7$, and $9$.
\end{enumerate}
If $\pi(\bullet) = 1$ then $9$ yields
$\m{c^1}\, y \geqslant_\kbo \m{c^1}\,\m{H} \, y >_\kbo \m{c^1} \, y$,
which leads to a contradiction. So $\pi(\bullet) = [1]$ holds.
Then $2$, $5$, $7$, and $9$ are expressed as follows:
\begin{align*}
{\bullet\,\talloblong}\,x &\geqslant_\kbo {\talloblong\bullet\bullet}~x
&
\m{c^1} \, z &\geqslant_\kbo {\circ} \, z
&
\bullet~\m{H} \, z &\geqslant_\kbo \m{c^1} \, z
&
{\bullet}~\m{c^1} \, y &\geqslant_\kbo^\pi \m{c^1}\m{H} \, y
\end{align*}
Their weight conditions impose
$w(\bullet) = w(\m{H}) = w(\m{c^1}) = w(\circ) = 0$.
Therefore, $\talloblong \succ \bullet$ follows from $2b$.
However, it contradicts the admissibility condition of the KBO.
\qed
\end{proof}

Similarly, we can show that \cref{cor:rp} is not applicable.

The following proposition states that the consequence does not change even
if we use a recursive path order (see e.g.~\cite[Section~6.4.1]{TeReSe} for
the definition).

\begin{proposition}
\label{prop:hydra-rpo}
There are no recursive path order $({\geqslant_\rpo}, {>_\rpo})$ and argument filter $\pi$ satisfying 
constraints $1\,\text{--}\,11$ and $1a\,\text{--}\,10b$ (even if
quasi-precedence is allowed).
\end{proposition}
\begin{proof}
Assume to the contrary that the constraints are satisfied.  
If $s \geqslant_\rpo^\pi t$ then all variables in $\hat{\pi}(t)$ occur in
$\hat{\pi}(s)$.  So we can deduce
\begin{enumerate}[(i)]
\item
$1 \in \pi(\bullet)$ from $4$,
\item
$1 \in \pi(\talloblong^\sharp)$ and $\pi(\circ) = [1]$
from $1b$,
\item
$1 \in \pi(\bullet^\sharp)$ and $\pi(\talloblong) = [1]$
from $2c$,
\item
$\pi(\bullet) = 1$ from $2$, and
\item
$\pi(\m{c}^1),\pi(H) \in \{[2], [1, 2]\}$
from $5$ and $6$.
\end{enumerate}
From these we can see that 9 cannot be satisfied. Contradiction.
\qed
\end{proof}

\section{Weighted Path Order and Max/Plus Interpretations}
\label{sec:wpo-maxplus}

In this section we discuss how the techniques of this paper can be used
for the weighted path order (WPO)~\cite{YKS15}, which is a key ingredient
of the termination tool \NaTT~\cite{YKS14}.

Let $\AA$ be an algebra with a non-empty carrier $A$. Assume that $\AA$
is equipped with an order pair $({\geqslant},{>})$. An $i$-th argument
position of $f$ is said to be
\begin{itemize}
\item
\emph{weakly simple} if 
$f_\AA(a_1, \ldots, a_i, \ldots, a_n) \geqslant a_i$ for all
$\seq{a}$; and
\item
\emph{strictly simple} if
$f_\AA(a_1, \ldots, a_i, \ldots, a_n) > a_i$ for all $\seq{a}$.
\end{itemize}
A \emph{partial status} $\pi$ maps an $n$-ary function symbol $f$ to a
subset of $\{ 1, \ldots, n \}$.\footnote{%
This is a simplified version of partial status, see~\cite{YKS15,TSSY20} for
more general versions.} If $\pi(f) = \{ i_1, \ldots, i_n \}$ with 
$i_1 \leqslant \cdots \leqslant i_n$ then
$\pi(f)(t_1, \ldots, t_m) = (t_{i_1}, \ldots, t_{i_n})$.
We say that $\AA$ is \emph{weakly (resp. strictly) $\pi$-simple}
if the argument position $i$ is weakly (resp. strictly) simple for all
function symbols $f$ and $i \in \pi(f)$. Finally,
$\AA$ is \emph{trivial} if the carrier $A$ is a singleton set.

\begin{definition}
Let $\pi$ be a partial status, $\AA$ an algebra, and $\succeq$ a precedence (a quasi-order on function symbols).
The \emph{weighted path order} $({\geqslant_\wpo}, {>_\wpo})$ is a pair of relations
on terms defined simultaneously: $s \geqslant_\wpo t$ if
\begin{enumerate}[1.]
\item
$s >_\AA t$, or
\item
$s \geqslant_\AA t$ and one of the following conditions holds:
\begin{enumerate}[a.]
\item
$s = f(\seq[m]{s})$ and $s_i \geqslant_\wpo t$ for some $i \in
\pi(f)$.
\item
$s = f(\seq[m]{s})$, $t = g(\seq[n]{t})$, $s >_\wpo t_j$ for all 
$j \in \pi(g)$, and
\begin{enumerate}[i.]
\item
$f \succ g$ or
\item
$f \succeq g$ and
$\pi(f)(\seq[m]{s}) \geqslant_\wpo^\lex \pi(g)(\seq[n]{t})$.
\end{enumerate}
\item
$s \in \VV$ and either $s = t$ or $t = g(\seq[n]{t})$, $\pi(g) = \emptyset$
and $g$ is least in $\succeq$.
\item
$s = f(\seq[m]{s})$, $t \in \VV$, $\AA$ is strictly simple with respect to
$\pi$, and for all function symbols $g$, either $f \succ g$ or $f \sim g$
and $\pi(g) = \emptyset$ holds.
\end{enumerate}
\end{enumerate}
Here $\succ$ is the strict part of $\succeq$ and
$\sim$ is the equivalence induced from $\succeq$. Moreover,
$\geqslant^\lex$ is the lexicographic extension
of $\geqslant$, see~\cite{YKS15}.
The relation $>_\wpo$ is defined by cases (1), (2a) and (2b)
with $\geqslant_\wpo^\lex$ replaced by $>_\wpo^\lex$.
\end{definition}

\begin{theorem}[\textnormal{\cite{YKS15,TSSY20}}]
Let $\pi$ be a partial status, $\succeq$ a precedence,
and $\AA$ an algebra that is well-founded, non-trivial, weakly $\pi$-simple and weakly monotone.
Then $({\geqslant_\wpo}, {>_\wpo})$ is a reduction pair.
\end{theorem}

In \cite{YKS15} it is shown that, if $\pi(f) = \emptyset$ and 
$f \succeq g$ for all function symbols $f$ and $g$, then
$({\geqslant_\wpo}, {>_\wpo})$ is identical to $({\geqslant_\AA},
{>_\AA})$.  
In such a setting weak and strict $\pi$-simplicities trivially holds.
So the weighted path
order subsumes all interpretation-based reduction pairs.

In order to use \cref{thm:combinability} with WPO, monotone and invariant
argument positions have to be identified. The next proposition can be used
for this purpose.

\begin{proposition}
\label{prop:wpo-monotone-invariant}
Let $\pi$ be a partial status, $\succeq$ a precedence, and $\AA$ an algebra.
For the weighted path order $({\geqslant_\wpo}, {>_\wpo})$,
an argument position $1 \leqslant i \leqslant n$ is monotone
if $i \in \pi(f)$;
similarly, $i$ is invariant if $i \notin \pi(f)$ and $i$ is an invariant position of $f$ with respect to $\AA$.
\end{proposition}

Note that the sufficient condition above for monotone positions is an
under-approximation, in particular when $({\geqslant_\wpo}, {>_\wpo})$ is
identical to $({\geqslant_\AA}, {>_\AA})$.
The invariance condition of $\AA$ cannot be dropped. To see it, 
consider the polynomial interpretation $\m{f}_\AA(x) = x + 1$ and the
partial status $\pi(\m{f}) = \emptyset$.  The first argument position of
$\m{f}$ is not invariant, as $\m{f}(\m{f}(x)) >_\wpo \m{f}(x)$ follows
from $\m{f}(\m{f}(x)) >_\AA \m{f}(x)$.

We also need to state criteria to detect weakly/strictly simple positions
of algebras to use them with WPO.
We begin with the standard matrix interpretation.

\begin{proposition}
\label{prop:standard}
Consider a standard matrix interpretation:
\[
f_\AA(\seq{\vec{x}}) = \vec{a} + A_1 \vec{x}_1 + \cdots + A_n \vec{x}_n
\]
The $i$-th argument position of $f$ is weakly simple if all the diagonal
entries of $A_i$ are positive, and it is strictly simple if in addition 
the first entry of $\vec{a}$ is positive.
\end{proposition}

\begin{proposition}
Consider an echelon-form matrix interpretation:
\[
f_\AA(\seq{\vec{x}}) = \vec{a} + A_1 \vec{x}_1 + \cdots + A_n \vec{x}_n
\]
The $i$-th argument position of $f$ is weakly simple if all the diagonal
entries of $A_i$ are positive, and it is strictly simple if in addition
$\vec{a}$ has a positive entry.
\end{proposition}

The \emph{max/plus interpretation}~\cite{YKS15} is often used 
not only as reduction pairs but also for constructing WPOs. It is an algebra whose
carrier is $\NN$ and interpretations are given by the form
\[
f_\AA(\seq{x}) = \max \{ a_0, b_1 (a_1 + x_1), \ldots, b_n (a_n + x_n) \}
\]
where $a_0, \seq{b} \in \NN$ and $\seq{a} \in \ZZ$. Its monotone/simple
positions are characterized as follows.

\begin{proposition}
\label{prop:maxplus}
Consider a max/plus interpretation $f_\AA$ of the above form.
The $i$-th argument position of $f$ is
weakly monotone for free; 
strictly monotone if $a_i \geqslant a_0$, $b_i > 0$,
and $b_j = 0$ for all other indices $j$;
weakly simple if $b_i > 0$ and $a_i \geqslant 0$;
and finally, strictly simple if $b_i > 0$ and $a_i > 0$.
\end{proposition}

\begin{table}[t]
\caption{Supplementary experiments on 1528 termination problems.}
\label{tbl:standard2}
\centering
\begin{tabular}{%
@{}l@{~~}
r@{~~}r@{~~}r@{~~}r@{~~}r@{~~}
r@{~~}r@{~~}r@{~~}r@{}}
\toprule
& $\m{M}$
& $\m{MM}$
& $\m{WM}$
& $\m{WE_1}$
& $\m{WE_2}$
& $\m{WS_2}$
& $\m{WM*WE_2}$
& $\m{WM*WS_2}$
& $\m{WE_2*WS_2}$
\\
proved
& 526
& 499
& 546
& 490
& 558
& 583
& 648
& 677
& 606
\\
\textit{timeout}
& \textit{13}
& \textit{13}
& \textit{54}
& \textit{30}
& \textit{44}
& \textit{72}
& \textit{77}
& \textit{112}
& \textit{124}
\\
\bottomrule
\end{tabular}
\end{table}

\begin{table}[t]
\caption{Supplementary experiments on 57 relative termination problems.}
\label{tbl:relative2}
\centering
\begin{tabular}{%
@{}l@{~~}
r@{~~}r@{~~}r@{~~}r@{~~}r@{~~}
r@{~~}r@{~~}r@{~~}r@{}}
\toprule
& $\m{M}$
& $\m{MM}$
& $\m{WM}$
& $\m{WE_1}$
& $\m{WE_2}$
& $\m{WS_2}$
& $\m{WM*WE_2}$
& $\m{WM*WS_2}$
& $\m{WE_2*WS_2}$
\\
proved
& 6
& 12
& 20
& 29
& 47
& 31
& 48
& 48
& 47
\\
\textit{timeout}
& \textit{0}
& \textit{0}
& \textit{0}
& \textit{0}
& \textit{0}
& \textit{0}
& \textit{0}
& \textit{0}
& \textit{0}
\\
\bottomrule
\end{tabular}
\end{table}

Finally, we report supplementary experiments on WPO.
\cref{tbl:standard2} and \cref{tbl:relative2} summarize the results,
see \cref{sec:experiments} for the settings of the experiments.
In the tables, $\m{M}$ denotes the max/plus interpretation, $\m{W}X$ the
weighted path order whose underlying interpretation is of the class
$X$, and $A * B$ tries all lexicographic combinations $AA$, $BB$, $AB$,
and $BA$ in this order.
For example, the column of $\m{WM*WS_2}$ in \cref{tbl:standard2} indicates that 
677 problems are proven terminating if the employed reduction pairs are
lexicographic combinations of WPO induced by the max/plus
interpretation and that induced by the 2-dimensional standard matrix
interpretation, while the tool run out of time for 124 problems.

As is seen in \cref{prop:maxplus}, a max/plus interpretation ($\m{M}$)
has at most one monotone position for each function symbol, so it has
a bad compatibility with the combinability criterion. 
Indeed, $\m{M}$ and $\m{MM}$ show that
the use of lexicographic combination may result in fewer proofs.
However, if we use it with WPO, the resulting reduction pair can have
more monotone positions (cf.~\cref{prop:wpo-monotone-invariant}), which is
more suited for lexicographic combination.  Indeed, combining $\m{WM}$ 
with other reduction orders is 
quite powerful,
see~\cref{tbl:standard2}.
Theoretically, the use of WPO with a class $X$ of algebras
increases termination proving power over just using $X$. However,
in our experiments this is not always observed because WPO gives
larger SMT encodings, which in turn lead to more timeouts.

\end{document}